\documentclass[10pt]{article}

\usepackage[round]{natbib}

\usepackage[frozencache,cachedir=.]{minted}

\usepackage[utf8]{inputenc} 
\usepackage[T1]{fontenc}    
\usepackage{hyperref}       
\usepackage{url}            
\usepackage{booktabs}       
\usepackage{amsfonts}       
\usepackage{nicefrac}       
\usepackage{microtype}      
\usepackage{xcolor}         
\usepackage{acro}
\usepackage{algorithm}
\usepackage{epstopdf}
\usepackage[noend]{algpseudocode}
\usepackage{amsmath}
\usepackage{amssymb}
\usepackage{amsthm}
\usepackage{bm}
\usepackage{bbm}
\usepackage{color}
\usepackage{comment}
\usepackage{dsfont}
\usepackage{enumitem}
\usepackage[utf8]{inputenc}
\usepackage{mathtools}
\usepackage{tabularx}
\usepackage{caption}
\usepackage{subcaption}
\usepackage{placeins}

\usepackage{fullpage}
\usepackage{wrapfig}
\allowdisplaybreaks

\usepackage{yfonts} 
\usepackage{wrapfig}
\usepackage{pifont}
\usepackage{cleveref}

\DeclareAcronym{PPL}{short=PPL, long=probabilistic programming language}
\DeclareAcronym{SED}{short=SED, long=sequential experimental design}
\DeclareAcronym{GP}{short=GP, long=Gaussian process}
\DeclareAcronym{MCMC}{short=MCMC, long=Markov chain Monte Carlo}
\DeclareAcronym{PDE}{short=PDE, long=partial differential equation}
\DeclareAcronym{MSE}{short=MSE, long=mean square error}

\newtheorem{theorem}{Theorem}

\newtheorem{lemma}{Lemma}
\newtheorem{definition}{Definition}
\newtheorem{proposition}{Proposition}
\newtheorem*{definition*}{Definition}

\crefname{assumption}{Assumption}{Assumptions}

\newcommand{\E}{\mathbb{E}}
\newcommand{\Var}{\mathbb{V}}
\newcommand{\Cov}{\mathbb{C}}

\DeclareMathOperator*{\argmax}{arg\,max}
\DeclareMathOperator*{\argmin}{arg\,min}

\newcommand*{\email}[1]{\bgroup\color{blue}\href{mailto:#1}{#1}\egroup}
\renewcommand*{\url}[1]{\bgroup\color{blue}\href{#1}{#1}\egroup}
\newcommand{\todo}[1]{\bgroup\color{red}TODO:~#1\egroup}

\algrenewcommand\algorithmicindent{1.em}%

\title{GaussED: A Probabilistic Programming Language for Sequential Experimental Design}
\author{Matthew A. Fisher$^1$, Onur Teymur$^{1,2}$, Chris. J. Oates$^{1,2}$ \\
\small $^1$Newcastle University, UK \\
\small $^2$Alan Turing Institute, UK }

\begin{document}

\maketitle

\begin{abstract}
\begin{abstract}
    Sequential algorithms are popular for experimental design, enabling emulation, optimisation and inference to be efficiently performed.
    For most of these applications bespoke software has been developed, but the approach is general and many of the actual computations performed in such software are identical.
    Motivated by the diverse problems that can in principle be solved with common code, this paper presents \texttt{GaussED}, a simple probabilistic programming language coupled to a powerful experimental design engine, which together automate sequential experimental design for approximating a (possibly nonlinear) quantity of interest in Gaussian processes models.
    Using a handful of commands, \texttt{GaussED} can be used to: solve linear partial differential equations, perform tomographic reconstruction from integral data and implement Bayesian optimisation with gradient data.
\end{abstract}
\end{abstract}

\section{Introduction} \label{sec: introduction}

This paper concerns the development of a \ac{PPL} for \ac{SED}.
A \ac{PPL} is an attempt to streamline the process of performing computation with a statistical model \citep{goodman2013principles}.
\ac{SED} is often associated with computational workflows that are complicated and cumbersome, as one is required to iterate between designing an experiment (to augment a dataset with a new datum) and performing inference for a specified quantity of interest (based on the augmented dataset).
Thus \ac{SED} is well-placed to benefit from the development of a high-level \ac{PPL}.
The research challenge here is to identify a class of statistical models that are sufficiently general to include important applications of \ac{SED}, while being sufficiently narrow to permit both inference and \ac{SED} to be efficiently and automatically performed.
This paper aims to address two important open problems in \ac{PPL} for \ac{SED}:
\begin{enumerate}
    \item[P1] automate \ac{SED} for \ac{GP} models with general nonlinear quantities of interest, in the setting of continuous linear functional data (e.g. function values, gradients, integrals);
    \item[P2] circumvent the requirement for the user to specify an acquisition function for \ac{SED}, in the spirit of \textit{AutoML} \citep{Hutter}.
\end{enumerate}
In limiting attention to the relatively narrow class of \ac{GP} models in P1, we aim to develop more powerful algorithms than would have been possible in a more general-purpose \ac{PPL}.
The setting of P1 includes \ac{SED} for the important tasks of \emph{emulating} computer models \citep{kennedy2001bayesian}, performing \emph{Bayesian optimisation} \citep{shahriari2015taking}, and running \emph{probabilistic numerical methods} \citep{hennig2015probabilistic}.
Bespoke \acp{PPL} have been developed for these individual tasks, but many of the actual computations performed in such software are identical.
Indeed, in \Cref{sec: demonstrations} we demonstrate how a single \ac{PPL} can: solve partial differential equations using a probabilistic numerical method, perform tomographic reconstruction from integral data, implement Bayesian optimisation with gradient data, and emulate a complex computer model.
Such a \ac{PPL} enables advances in computational methodology to be immediately brought to bear on diverse application areas where \ac{SED} is performed.

Existing \acp{PPL} for \ac{SED} require the user to specify an \emph{acquisition function}, which is used to select the next experiment and serves to control the exploration-exploitation trade-off.
Unfortunately, the process of determining an effective acquisition function requires domain expertise and, while several choices have been documented in the literature \citep[see e.g.][for acquisition functions in Bayesian optimisation]{wilson2018acquisition}, many problems that fall into the setting of P1 have not received such detailed treatment.
In removing the technical burden of prescribing the acquisition function in P2, we may sacrifice a degree of performance relative to dedicated software for tasks such as Bayesian optimisation, for which bespoke acquisition functions have been developed.
However, empirical results in this paper suggest that the loss of performance may be modest, and in turn we are able to considerably expand the applicability of the \ac{PPL}.

\subsection{Our Contribution} \label{subsec: contribution}

In this paper we present \texttt{GaussED}, a simple \ac{PPL} coupled to a powerful \emph{experimental design engine} for performing \ac{SED} in the nonparametric \ac{GP} context.
\texttt{GaussED} achieves the aims P1 and P2, just outlined.
To achieve P1, and to ensure that \texttt{GaussED} can handle data arising from general continuous linear functionals, we present a rigorous probabilistic treatment of conditioning for \acp{GP}.
This enables us to, for example, prevent attempts to condition on a derivative that does not exist under the \ac{GP} model.
To achieve P2 and circumvent the user-specification of an acquisition function, we adopt a classical but surprisingly overlooked decision-theoretic approach to \ac{SED}, which requires only the quantity of interest and a loss function to be specified.
The loss function quantifies the loss incurred when the true quantity of interest is approximated, a notion that is meaningful in the applied context and comparatively straightforward to elicit.
The computational backend for \texttt{GaussED} comprises a spectral \ac{GP}, a reparametrisation trick, and stochastic optimisation over the experimental design set.

\subsection{Related Work} \label{subsec: related work}

Several general-purpose \acp{PPL} have been developed for Bayesian parameter inference in parametric models \citep[e.g.][]{wood2014new,carpenter2017stan,bingham2019pyro}, often based on Markov chain Monte Carlo or variational approximations in the backend.
Specialised \acp{PPL} have been developed for inferring parameters that minimise a predictive loss \citep[e.g. using neural networks;][]{paszke2019pytorch}, often based on automatic differentiation and stochastic gradient descent.
For inference in nonparametric models, specialised \acp{PPL} have been developed for \ac{GP} models \citep[e.g.][]{rasmussen2010gaussian,matthews2017gpflow}, including for numerical applications \citep{probnum}.

The combination of \ac{PPL} and \ac{SED} for general parametric models has received attention in \citet[Chapter 11]{rainforth2017automating} and \citet{ouyang2016practical,kandasamy2018myopic}, who provided a high-level syntax for Bayesian \ac{SED}.
Several application-specific \ac{PPL} have been also been developed for \ac{SED} in parametric models \citep[e.g.][]{liepe2013maximizing}.
The focus of much of the research involving parametric models centres around the computational challenge of conditioning random variables on observed data, a problem that is often difficult \citep{olmedo2018conditioning}.

\ac{SED} for nonparametric models has received considerable attention in the context of Bayesian optimisation; see the review of \citet{shahriari2015taking}.
However, existing \acp{PPL} are specialised to this single task.
More closely related to the present paper, \cite{emukit2019} developed a \ac{PPL} called \texttt{Emukit}, in which computer model emulation, Bayesian optimisation, and a number of probabilistic numerical methods are automated.
However, \texttt{Emukit} focuses on function-value data as opposed to general continuous linear functionals (c.f. P1) and requires the user to specify a suitable acquisition function (c.f. P2). 

\paragraph{Outline:}{ The remainder of the paper is structured as follows: 
\Cref{sec: methods} presents a detailed technical description of \texttt{GaussED}. 
\Cref{sec: demonstrations} described the syntax of \texttt{GaussED} and presents diverse applications of \ac{SED}, for which bespoke code had previously been developed but whose automation is essentially trivial using \texttt{GaussED}.
The potential and limitations of \texttt{GaussED} are summarised in \Cref{sec: discussion}.
}

\section{Methodology} \label{sec: methods}

This section presents the statistical and computational methodology used in \texttt{GaussED}.
First, in \Cref{subsec: notation}, the notation and mathematical set-up are introduced.
The elements of \ac{SED} are outlined in \Cref{subsec: bayesian experimental design} and a classical, but surprisingly overlooked, approach to \ac{SED} is presented in \Cref{subsec: decision theoretic approach}.
This decision-theoretic approach circumvents the requirement to specify an acquisition function and, moreover, enables state-of-the-art stochastic optimisation to be employed in \ac{SED}, as explained in \Cref{subsec: stochastic optimisation,subsec: spectral}.
The hyperparameters of the \ac{GP} model are estimated online during \ac{SED}, as explained in \Cref{subsec: hyperparam}.

\subsection{Notation and Set-Up} \label{subsec: notation}

Let $\mathcal{F}$ be a normed vector space of real-valued functions on some domain $\mathcal{X}\subseteq \mathbb{R}^d$.
The problems that we consider involve a latent function $\mathsf{f} \in \mathcal{F}$, associated with a high computational cost, and the task is to approximate a (possibly nonlinear) quantity of interest $q(\mathsf{f})$ using \ac{SED}. 
The experiments are represented\footnote{The focus of this paper is on data that are \emph{exactly} observed, and as such we do not introduce a measurement error model. 
Gaussian errors can be handled in \texttt{GaussED} by building measurement error into the \ac{GP} covariance model.} as continuous linear functionals $\delta : \mathcal{F} \rightarrow \mathbb{R}$ and may, for example, include pointwise evaluation $\delta(\mathsf{f}) = \mathsf{f}(x)$ of the latent function $\mathsf{f}$ at a specified location $x \in \mathcal{X}$, pointwise evaluation of a gradient, or evaluation of an integral, such as a Fourier transform. 
A limited computational budget motivates the careful selection of informative experiments $\delta_1,\dots,\delta_n$.
\ac{SED} is often preferred\footnote{Sequential design is known to be near-optimal under \emph{adaptive submodularity} \citep{golovin2011adaptive}.} over \textit{a priori} experimental design, since it allows data $\delta_1(\mathsf{f}), \dots, \delta_{n-1}(\mathsf{f})$, which have already been observed, to inform the design of the next functional $\delta_{n}$.

Bayesian statistics provides a general framework in which \ac{SED} can be performed.
To this end, let $(\Omega,\mathcal{S},\mathbb{P})$ be a probability space and consider a random variable $f : \Omega \rightarrow \mathcal{F}$.
This serves as a statistical model for the latent $\mathsf{f}$, and encodes \emph{a priori} knowledge, such as the smoothness of $\mathsf{f}$. 
To notate the distribution of $f$, we first define the \emph{pre-image} of a set $B \subseteq \mathcal{F}$ as $f^{-1}(B) := \{\omega \in \Omega : f(\omega) \in B\}$ and we let $f_\# \mathbb{P}$ denote the \emph{pushforward} of $\mathbb{P}$ through $f$; i.e. the probability distribution on $\mathcal{F}$ that assigns, to each Borel set $B \subseteq \mathcal{F}$, the mass $f_\# \mathbb{P}(B) := \mathbb{P}(f^{-1}(B))$.
The distribution of $f$ will be denoted $\mathbb{P}_f := f_\# \mathbb{P}$ in the sequel.
Our presentation allows for general priors for $f$ until \Cref{subsec: spectral}, at which point we will assume $f$ is a \ac{GP}.
Throughout we adopt the convention that $\mathsf{f}$ refers to the latent function of interest, $f$ is a random variable model for $\mathsf{f}$, and $\mathrm{f}$ is a generic element of the set $\mathcal{F}$.

\subsection{Sequential Experimental Design} \label{subsec: bayesian experimental design}

\ac{SED} iterates between designing an experiment $\delta_n$, to augment a dataset with a new datum $\delta_n(\mathsf{f})$, and performing inference for a specified quantity of interest, based on the augmented dataset $\bm{\delta}_n(\mathsf{f}) := (\delta_1(\mathsf{f}), \dots, \delta_n(\mathsf{f}))^\top$. 
Let $\mathcal{D} \subseteq \mathcal{F}'$ indicate the \textit{design set}, where $\mathcal{F}'$ is the topological dual space of $\mathcal{F}$, containing the continuous linear functionals on $\mathcal{F}$.
The design set $\mathcal{D}$ will depend on the problem at hand, and contains only the experiments that can actually be performed.
At iteration $n$, \ac{SED} selects an experiment $\delta_n$ from the design set in order that an \emph{acquisition function} is maximised\footnote{To avoid pathological cases, in this paper the existence of a (not necessarily unique) maximum is always assumed.}:
\begin{equation} 
\delta_n \in \argmax_{\delta \in \mathcal{D}} A(\delta ; \mathbb{P}_f ,  \bm{\delta}_{n-1}(\mathsf{f})) \label{eq: acq}
\end{equation}
The role of the acquisition function $A$ is to control the exploration-exploitation trade-off, but the computational convenience of computing \eqref{eq: acq} is also important.
Much research has been dedicated to exploring choices for $A$, and the statistical and computational properties of the associated sequence $(\delta_n)_{n=1}^\infty$.
Specific applications, where interest is not necessarily in $\mathsf{f}$ but rather a derived quantity of interest $q(\mathsf{f})$, have developed bespoke acquisition functions that balance computational cost with accurate approximation of the quantity of interest, in particular in Bayesian optimisation \citep[see Table 1 in][]{wilson2018acquisition}.
This presents a major problem (P2) for the development of a general purpose \ac{PPL} for \ac{SED}, since in general we cannot expect a user to specify a suitable acquisition function for the problem at hand.

As a first step toward solving P2, we consider a Bayesian approach to the design of an acquistion function.
To this end, let $\mathbb{P}_f(\cdot | \bm{\delta}_n(\mathsf{f}) )$ denote the conditional distribution (or \emph{posterior}) of $f$ obtained by setting the values $\bm{\delta}_n(f)$ equal to the observed data $\bm{\delta}_n(\mathsf{f})$.
From a mathematical perspective, the proper construction of a conditional distribution for an infinite-dimensional random variable $f$ is non-trivial; we suppress further discussion in the main text but refer the reader to \Cref{app: math prelim} for full mathematical detail. 
A Bayesian approach to the design of an acquisition function is then to let $U : \mathbb{R}^{n-1} \times \mathbb{R} \rightarrow \mathbb{R}$ be a \textit{utility} function, to be specified, and to seek an experiment for which the current expected utility 
\begin{equation} \label{eq: bayes decision}
A(\delta; \mathbb{P}_f ,  \bm{\delta}_{n-1}(\mathsf{f})) = \textstyle \int U(\bm{\delta}_{n-1}(\mathsf{f}), \delta(\mathrm{f}) ) \; \mathrm{d}\mathbb{P}_f(\mathrm{f} | \bm{\delta}_{n-1}(\mathsf{f}) )
\end{equation}
is maximised.
The utility $U(\bm{\delta}_{n-1}(\mathsf{f}), \delta(\mathrm{f}) )$ represents the value to the user of observing the datum $\delta(\mathrm{f})$.
Thus the design of an acquisition function can be reduced to the design of a utility function.
A popular default choice for $U$ is the \textit{information gain} \citep{lindley1956measure}
\begin{align}
\text{KL}( \; \mathbb{P}_f(\cdot | \bm{\delta}_{n-1}(\mathsf{f}), \delta(\mathrm{f}) ) \; \| \; \mathbb{P}_f(\cdot | \bm{\delta}_{n-1}(\mathsf{f}) ) \; ) , \label{eq: EU}
\end{align}
which quantifies the extent to which observation of the datum $\delta(\mathrm{f})$ changes \textit{a posteriori} belief; here $\text{KL}$ denotes the Kullback--Leibler divergence.
For related approaches and discussion see the recent survey in \cite{kleinegesse2021gradientbased}.
However, in the setting where data are exactly observed, the two distributions in \eqref{eq: EU} will be mutually singular and the Kullback--Leibler divergence will not exist.
This renders information-based acquisition functions such as \eqref{eq: EU} unsuitable for our \ac{PPL}.
Instead, we propose to revisit a classical but often overlooked idea from experimental design, next.

\subsection{A Decision-Theoretic Approach} \label{subsec: decision theoretic approach}

A general approach to construction of a utility $U$ is provided by Bayesian decision theory in the \textit{parameter inference} context\footnote{The decision-theoretic approach was advocated by \citet[][Section 2.5]{berger1985decision}, who wrote ``\textit{better inferences can often be done with the aid of decision-theoretic machinery and inference losses}''. }.
Let $L : \mathcal{F} \times \mathcal{F} \rightarrow \mathbb{R}$ denote the \textit{loss} $L(\mathrm{f},\mathrm{g})$ when estimating the function (or \textit{parameter}) $\mathrm{f}$ by $\mathrm{g}$.
Then we can take $U$ to be the negative Bayes' risk
\begin{align}
    \textstyle - \min_{\mathrm{g} \in \mathcal{F}} \int L(\mathrm{g},\mathrm{g}') \; \mathrm{d}\mathbb{P}_f(\mathrm{g}' | \bm{\delta}_{n-1}(\mathsf{f}) , \delta(\mathrm{f}) ) , \label{eq: U bayes risk}
\end{align}
which corresponds to the negative expected loss when the Bayes decision rule $\mathrm{g}$ is used.
Compared to an acquisition function or a utility function, it can be more straightforward to specify a suitable loss function $L$, since no consideration of the design set is required.
Although appealing in terms of its generality, the presence of the optimisation over $\mathrm{g}$ has historically rendered this utility unappealing from a computational viewpoint, and motivated more convenient choices, such as \eqref{eq: EU}, that have since become canonical \citep[see the survey in][]{chaloner1995experimental}.
However, we argue that the presumed intractability of loss-based utilities might need to be revisited in light of modern and powerful stochastic optimisation techniques.
Indeed, for loss functions of the form $L(\mathrm{f},\mathrm{g}) = \|q(\mathrm{f}) - q(\mathrm{g})\|^2$, indicating that one has a quantity of interest $q(\mathrm{f})$ taking values in a normed space\footnote{A focus on squared error loss is only a mild restriction, since we are free to re-parametrise the quantity of interest $q$ as $t \circ q$, where $t$ is an injective map (to ensure that information is not lost). Through careful selection of $t$ we may formulate the \ac{SED} task in a setting where squared error loss is appropriate for the task at hand.}, under mild conditions \eqref{eq: U bayes risk} is equal to
\begin{align}
    \textstyle - \frac{1}{2} \iint  L(\mathrm{g},\mathrm{g}') \; \mathrm{d}\mathbb{P}_f(\mathrm{g} | \bm{\delta}_{n-1}(\mathsf{f})  , \delta(\mathrm{f})  ) \; \mathrm{d}\mathbb{P}_f(\mathrm{g}' | \bm{\delta}_{n-1}(\mathsf{f}) , \delta(\mathrm{f}) ) . \label{eq: rewrite U}
\end{align}
The required regularity conditions and a formal proof are contained in \Cref{app: decision theory}.
At first glance it is unclear why this observation is helpful, since we have replaced an optimisation problem with an integration problem, and integration is typically \textit{more} difficult than optimisation.
However, this formulation turns the experimental design problem to find $\delta_n$ into a double expectation and, if the design set $\mathcal{D}$ has enough structure for calculus, then gradient-based stochastic optimisation can be applied. 

The restriction to squared error loss is not as limited as it may first appear, since one has the freedom to specify the quantity of interest $q(\mathsf{f})$ in such a way that application of squared error loss to $q(\mathsf{f})$ captures salient aspects of the task at hand.
Concrete examples of this are provided in \Cref{subsec: tomographic reconstruction}.

\subsection{Stochastic Optimisation} \label{subsec: stochastic optimisation}

Following this decision-theoretic approach, an acquisition function is obtained in expectation form by plugging \eqref{eq: rewrite U} into \eqref{eq: bayes decision} and applying the law of total probability, producing
\begin{align}
    A(\delta; \mathbb{P}_f ,  \bm{\delta}_{n-1}(\mathsf{f})) = \textstyle - \frac{1}{2} \iint L(\mathrm{g},\mathrm{g}') \; \mathrm{d}\mathbb{P}_f(\mathrm{g}' | \bm{\delta}_{n-1}(\mathsf{f})  , \delta(\mathrm{g})  ) \; \mathrm{d}\mathbb{P}_f( \mathrm{g} | \bm{\delta}_{n-1}(\mathsf{f})  ). \label{eq: final A} 
\end{align}
This acquisition function does not permit a closed form in general.
Several numerical methods have been proposed for maximisation of acquisition functions in the literature, including Bayesian optimisation \citep{overstall2017design,kleinegesse2019implicit}, non-gradient based Monte-Carlo methods, and approximation strategies. 
Similar to the approach\footnote{\cite{wilson2018acquisition} performed a reparametrisation trick by restricting attention to acquisition functions that depend on the \ac{GP} only at a finite number of locations in the domain $\mathcal{X}$; in contrast, this paper exploits a spectral approximation of the \ac{GP}, described in \Cref{subsec: spectral}.} of \cite{wilson2018acquisition}, here we consider the use of stochastic optimisation techniques \citep{robbins1951stochastic} for selecting an experiment $\delta$ for which \eqref{eq: final A} is approximately maximised.
For an overview of stochastic optimisation, see \cite{kushner2003stochastic, ruder2016overview}. 
First we perform a \emph{reparametrisation trick} \citep{kingma2013auto}, expressing
\begin{align}
g' \sim \mathbb{P}_f( \cdot | \bm{\delta}_{n-1}(\mathsf{f}) , \delta(\mathrm{g}))  \Leftrightarrow g' = \eta(\omega ; \mathbb{P}_f,  \bm{\delta}_{n-1}(\mathsf{f}) , \delta(\mathrm{g})), \; \omega \sim \mathbb{P},  \label{eq: reparam trick}
\end{align}
using a deterministic transformation $\eta$ of a random variable $\omega$ that is $\delta$-independent.
\Cref{subsec: spectral}, below, details how we applied the reparametrisation trick to a \ac{GP} model.
Now, suppose further that the elements of the design set can be parametrised as $\mathcal{D} = \{\delta_z\}_{z \in \mathbb{R}^m} \subseteq \mathcal{F}$.
Assuming sufficiently regularity for the following calculus to be well-defined, an unbiased estimator of the gradient of the acquisition function is
\begin{align*}
     \textstyle \frac{\partial}{\partial z_i} A(\delta_z; \mathbb{P}_f ,  \bm{\delta}_{n-1}(\mathsf{f})) \approx \textstyle - \frac{1}{2} \frac{1}{NM} \sum_{i=1}^N \sum_{j=1}^M \frac{\partial}{\partial z_i} L(g_i, \eta( \omega_{ij} , \mathbb{P}_f , \bm{\delta}_{n-1}(\mathsf{f}) , \delta_z(g_i)) ) , \label{eq: unbias est}
\end{align*}
where the $g_i$ are independent random variables with distribution $\mathbb{P}_f(\cdot | \bm{\delta}_{n-1}(\mathsf{f}))$ and the $\omega_{ij}$ are independent random variables with distribution $\mathbb{P}$.
This is an instance of \emph{nested Monte Carlo}. 
The optimal balance between $N$ and $M$ for a fixed computational budget is discussed in \cite{rainforth2018nesting}; for a continuously differentiable gradient, an optimal choice\footnote{The values $M = 9$, $N=9^2$, were used for all experiments we report, being among the smallest values for which stochastic optimisation was routinely successful.} is $N\propto M^2$. 

\texttt{GaussED} exploits state-of-the-art spectral \acp{GP} to perform the reparametrisation trick, as presented next.

\subsection{Spectral Approximation of GPs} \label{subsec: spectral}

Up to this point our discussion applied to general statistical models $\mathbb{P}_f$ for the latent function $\mathsf{f}$.
In the remainder \acp{GP} will be used, since they facilitate closed form conditional distributions, as appearing in \eqref{eq: final A}.
The purpose of this section is twofold; to briefly introduce \acp{GP} and to describe how the reparametrisation trick can be performed.

A random variable $f$ taking values in a normed vector space $\mathcal{F}$ is \emph{Gaussian} if, for every continuous linear functional $\delta:\mathcal{F}\rightarrow\mathbb{R}$, the random variable $\delta(f)$ is a Gaussian on $\mathbb{R}$; see Definition 2.41 in \citet{sullivan2015uncertainty}.
It follows that the statistical properties of a \ac{GP} are characterised by its mean function $\mu(x) \coloneqq \E[f(x)]$, $x \in \mathcal{X}$, and covariance function $k(x,y) \coloneqq \Cov[f(x),f(y)]$, $x,y \in \mathcal{X}$, and we write $f\sim \mathcal{GP}(\mu,k)$. 
\acp{GP} admit conjugate inference, meaning that for a continuous linear functional $\delta\in \mathcal{D}$, the conditional distributions $\mathbb{P}_f(\cdot | \delta(\mathsf{f}))$ are also Gaussian, with mean and covariance functions that can be computed in closed form; see \Cref{subsec: disintegration of gaussian}.

For the reparametrisation trick, we aim to write a \ac{GP} as a deterministic transformation $f = \eta(\omega)$ of a random variable $\omega$, such that the distribution of $\omega$ does not depend on $\mu$ or $k$.
However, being a nonparametric statistical model, an infinite-dimensional $\omega$ will in general be required.
This motivates the use of an accurate finite-dimensional approximation of a \ac{GP} at the outset, i.e. for the prior $\mathbb{P}_f$.
A truncated Karhunen--Loeve expansion  \citep[see e.g. Theorem 11.4 in][]{sullivan2015uncertainty} in principle provides such a transformation, however this requires computation of the eigenfunctions of $k$, and linear functionals thereof, which will in general be difficult. The solution adopted in \texttt{GaussED} is to use the finite-rank approximation to isotropic \acp{GP} introduced in \cite{solin2019hilbert}: $f = \eta(\omega) = \mu + \sum_{i=1}^m \omega_i \phi_i$,
where the coefficients $\omega_i \sim \mathcal{N}(0,s(\sqrt \lambda_i))$ are independent, $s$ is the \emph{spectral density} of $k$, and $(\phi_i , \lambda_i)$ are the pairs of eigenfunctions and eigenvalues of the Laplacian $\Delta$ over the domain $\mathcal{X}$; see \Cref{app: eigenfunctions} for detail. 
The approximation converges as $m \rightarrow \infty$, with small values of $m$ often sufficient for accurate approximation; see \cite{riutortmayol2020practical}.
\texttt{GaussED} puts the user in control of $m$, since $m$ is the principal determinant of computational complexity in the experimental design engine, aside from the computations involving the latent function $\mathsf{f}$ itself.

\subsection{Hyperparameter Estimation} \label{subsec: hyperparam}

To this point we assumed that a \ac{GP} model can be specified at the outset.
In reality one is usually prepared only to posit a parametric class of \acp{GP} whose parameters (called \emph{hyperparameters}) are jointly estimated.
In \texttt{GaussED} the hyperparmaters of the \ac{GP} are estimated at each iteration $n \geq n_0$ of \ac{SED}, using the available dataset $\bm{\delta}_n(\mathsf{f})$, after an initial number $n_0 \in \mathbb{N}$ of data have been observed. 
Maximum likelihood estimation is employed, facilitated using automatic differentiation and Adam \citep{kingma2014adam}. 
The role of $n_0$ is to guard against over-confident inferences, since maximum likelihood tends to overfit when the dataset is small; see e.g. Chapter 5 of \citet{rasmussen2006}. In \texttt{GaussED}, the default value is taken as $n_0 = 10$.

This completes our description of \texttt{GaussED}. 
Our attention turns, next, to demonstrating and assessing its capabilities.

\section{Demonstration} \label{sec: demonstrations}

The aims of this section are to validate \texttt{GaussED} and to highlight the diverse and non-trivial applications that can be tackled.
\texttt{GaussED} is based on \texttt{Python} and utilises the automatic differentiation capabilities of \texttt{Pytorch} \citep{paszke2019pytorch}. Source code and documentation for \texttt{GaussED} can be downloaded from \url{https://github.com/MatthewAlexanderFisher/GaussED}.

Full details for each of the following examples are provided in \Cref{sec: experiment details}. 
An investigation into the sensitivity of the computational methodology to initial conditions, the choice of stochastic optimisation method, and the number of basis functions $m$, can be found in \Cref{sec: evaluating Gaussed}.

\subsection{Probabilistic Solution of PDEs} \label{subsec: heat equation}

Our first example concerns the probabilistic numerical solution of Poisson's equation with Dirichlet boundary conditions; the intention is to validate our methodology on a problem that is well-understood.
\ac{SED} for such problems was investigated with bespoke code in \cite{cockayne2016probabilistic}.
The PDE we consider is defined on $\mathcal{X} = [-1,1]^2$ and takes the form
\begin{alignat*}{2}
    \Delta \mathsf{f}(x) &= \mathsf{g}(x), &&\qquad x \in \mathcal{X}, \\
    \mathsf{f}(x) &= 0, &&\qquad x \in \partial\mathcal{X}.
\end{alignat*}
Our quantity of interest is the solution $\mathsf{f}$ and the black-box source $\mathsf{g}$ is assumed to be associated with a computational cost, so that numerical uncertainty quantification is required. 
For this demonstration we simply took
$$
\mathsf{g}(x) = -320| x_1^3 \exp\{-(3.2 x_1)^2 - (10x_2 - 5)^2\}| 
$$
as a test bed.
The latent $\mathsf{f}$ was modelled as a \ac{GP} $f$ with mean zero and Mat\'ern covariance with smoothness parameter $\nu = 3 + \frac{1}{2}$, ensuring the corresponding GP samples are almost surely contained in $C^3(\mathcal{X})$, implying the evaluations of the Laplacian of $f$ are continuous linear functionals (see \Cref{app subsec: sufficient matern}).
The design set $\mathcal{D}$, parameterised by $x\in \mathcal{X}$, consists of functionals of the form $\delta(\mathsf{f}) = \Delta \mathsf{f}(x)$. 
It is known that an optimal experimental design in this case is \textit{space filling} \citep{wendland_2004,Novak2010FillDistance}, as quantified by the \textit{fill distance}
\begin{equation*}
    \text{FD}(\{x_i\}_{i=1}^n, \mathcal{X}) \coloneqq \sup_{x \in \mathcal{X}} \Big\{ \min_{i \in \{1,\ldots,n\}} \|x - x_i\| \Big\} ,
\end{equation*}
and this fact will be used to validate \texttt{GaussED}.
The syntax of \texttt{GaussED} is demonstrated in \Cref{fig: syntax example}, and consists of specifying a covariance function (\texttt{k}), a quantity of interest (\texttt{qoi}), an observation model (\texttt{obs}), here the Laplacian (\texttt{Laplace}), a loss function (\texttt{loss}), a design (\texttt{d}) initialised with an \texttt{initial\_design}, and an acquisition function (\texttt{acq}). 
\texttt{BayesRisk} is the default acquisition function from \eqref{eq: final A}, but \texttt{GaussED} retains the capability for alternative acquisition functions in the event that they can be user-specified. 
The experiment object (\texttt{experiment}) then collates these objects together to perform $n=150$ iterations of \ac{SED}, optimising hyperparameters as specified in \Cref{subsec: hyperparam}. 

\begin{figure}
\centering
\begin{minted}[
frame=lines,
framesep=5mm,
baselinestretch=1.2,
fontsize=\footnotesize,
%linenos
]{python}
k = MaternKernel(3, dim=2)
qoi = SpectralGP(k)
obs = Laplacian(qoi)
loss = L2(qoi)

d = EvaluationDesign(obs, initial_design)
acq = BayesRisk(qoi, loss, d)

experiment = Experiment(obs, laplace_f, d, acq)
experiment.run(n=150)
\end{minted}
\caption{Example syntax for \texttt{GaussED}.
} \label{fig: syntax example}
\end{figure}

Results are shown in \Cref{fig: poisson design 2D} and required only the $8$ lines of code shown in \Cref{fig: syntax example}.
The number of basis functions used was $m=30^2$, we computed $n_0 = 10$ iterations of SED before beginning hyperparameter optimisation and a total of $9$ CPU hours were invested to ensure that all $n=150$ instances of stochastic optimisation converged. 
The fill distance is lower-bounded by $\Theta(n^{-1/2})$, and \Cref{fig: fill distance} demonstrates that this optimal rate is empirically achieved by \texttt{GaussED}.
This validates our approach to \ac{SED}.

\begin{figure}[t!]
   \centering
    \begin{subfigure}[b]{0.25\linewidth}
    \includegraphics[width=\textwidth]{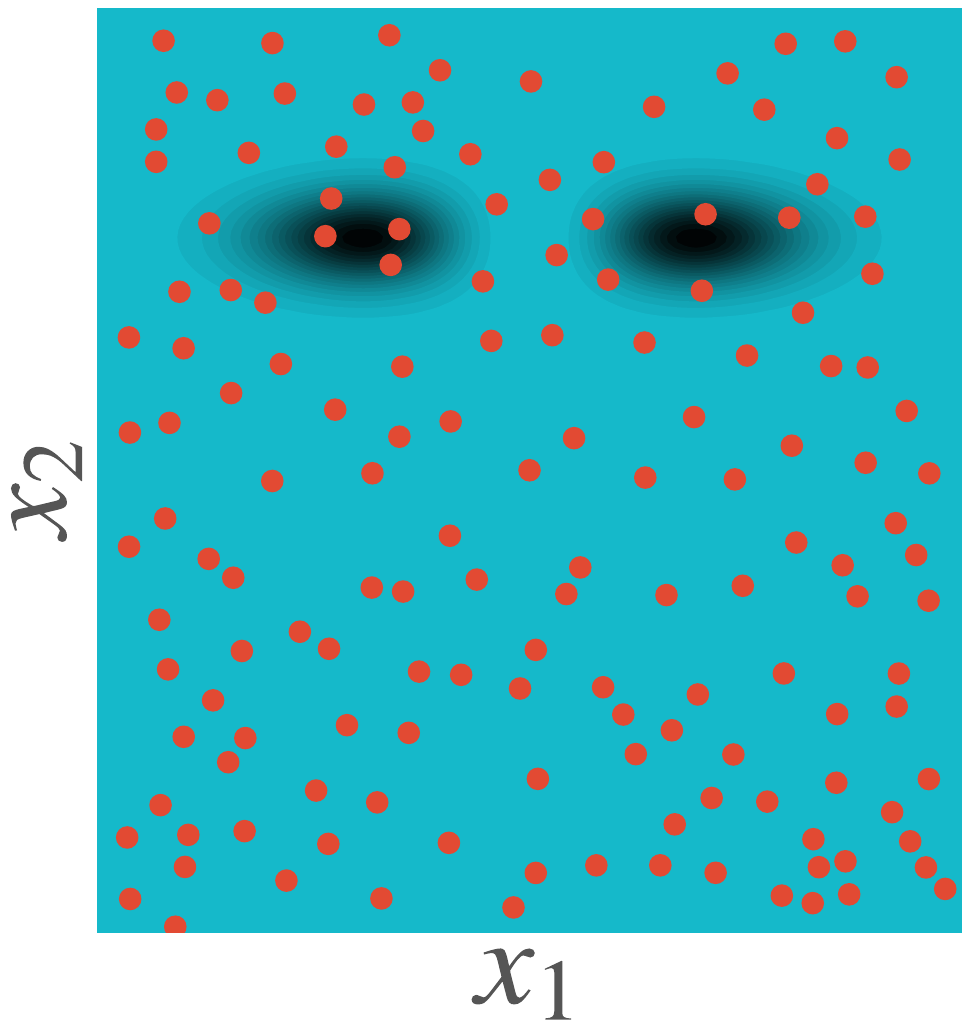} 
    \caption{}
    \end{subfigure}
   \begin{subfigure}[b]{0.25\linewidth}
    \includegraphics[width=\textwidth]{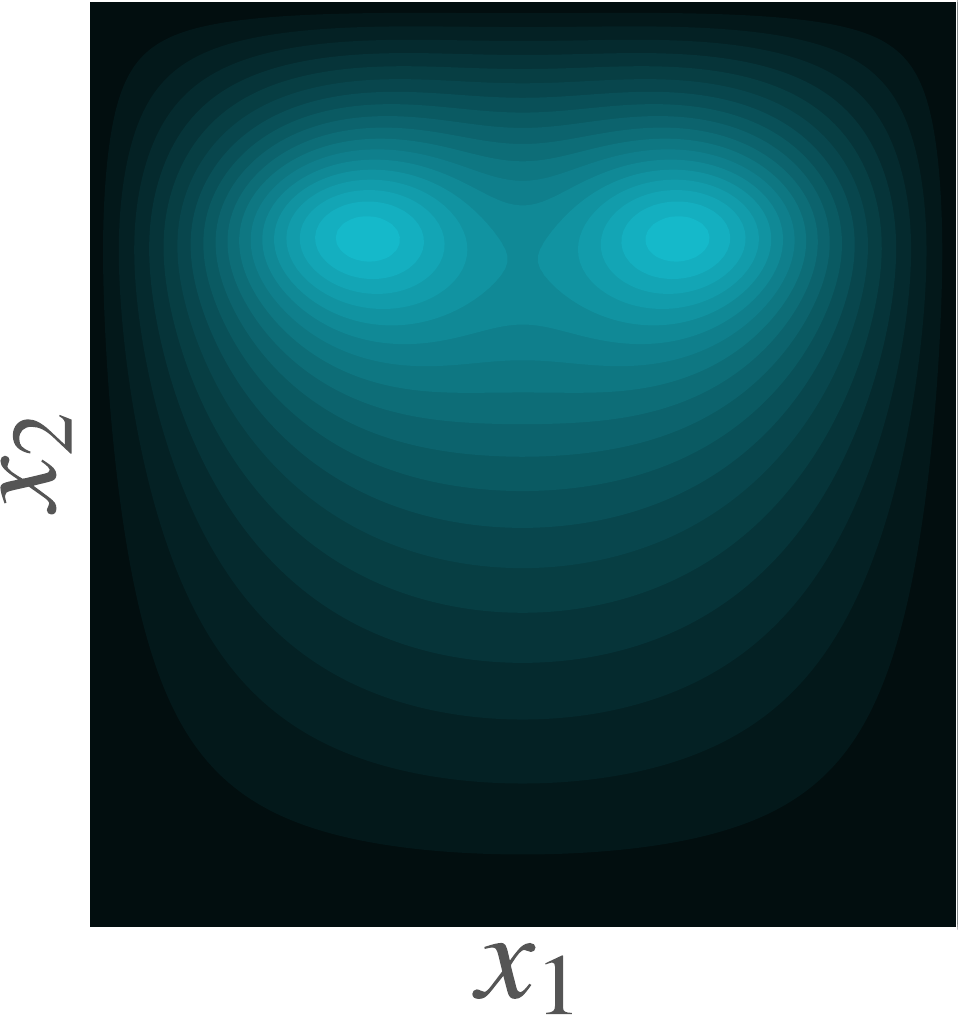} 
    \caption{}
    \end{subfigure}
    \begin{subfigure}[b]{0.25\linewidth} 
    \includegraphics[width=\textwidth]{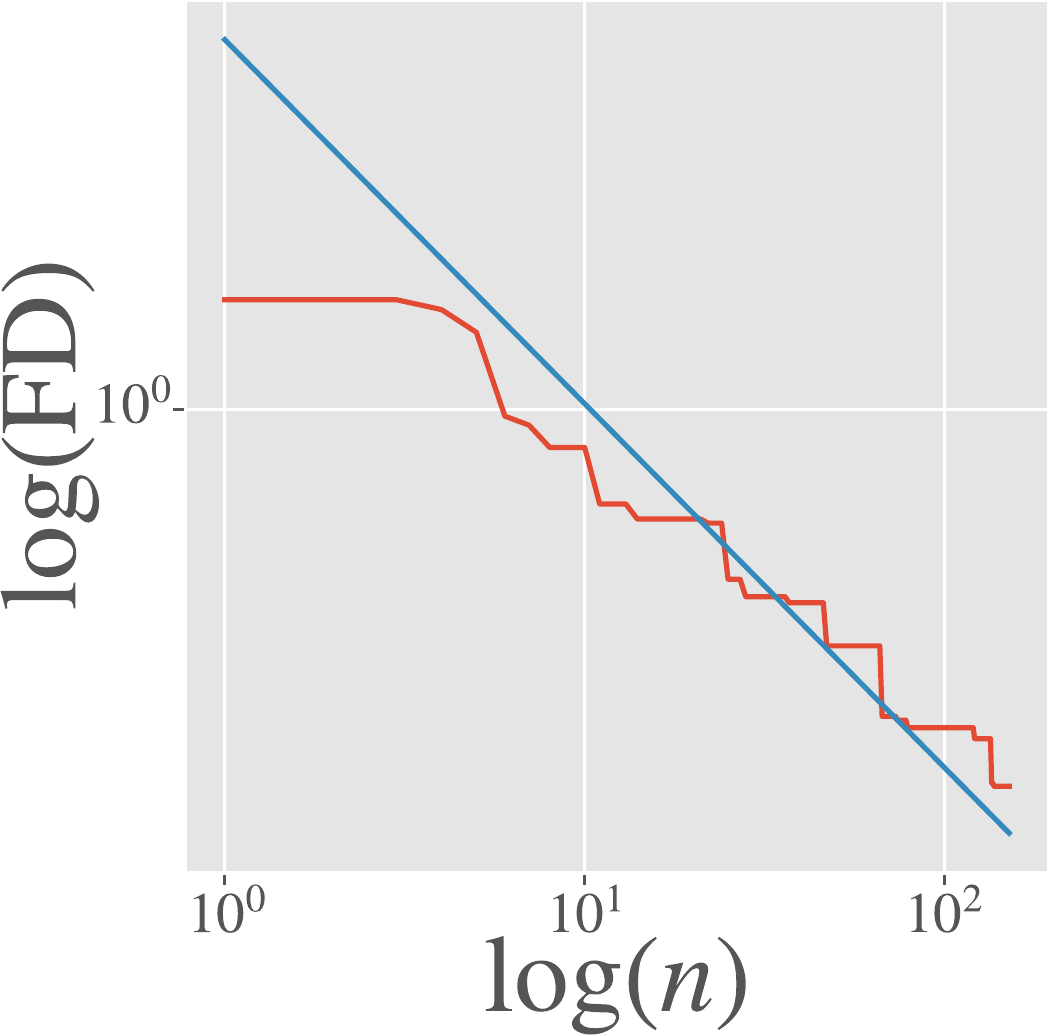}
    \caption{} \label{fig: fill distance}
    \end{subfigure}
    
\caption{
\textit{Probabilistic Solution of PDEs:}
    (a) Source term $\mathsf{g}$ with design points (red) determined by \ac{SED} overlaid.
    (b) Mean of $f | \bm{\delta}_n(\mathsf{f})$, the posterior obtained using \ac{SED}.
    (c) Fill distance (FD; red) versus the number $n$ of iterations in \ac{SED}, with theoretical optimal slope $-\frac{1}{2}$ (blue) displayed.
}
\label{fig: poisson design 2D}
\end{figure}

\subsection{Tomographic Reconstruction} \label{subsec: tomographic reconstruction}

Our next example is tomographic reconstruction from x-ray data \citep{mersereau1974reconstruction}.
The aim is to reconstruct a latent function $\mathsf{f}:\mathcal{X} \rightarrow\mathbb{R}$, where $\mathcal{X} = [-1,1]^d$, using line-integral data of the form
\begin{equation*}
    \delta(\mathsf{f}) = \textstyle \int_a^b \mathsf{f}(r(t))\left|r'(t)\right| \, \mathrm{d}t,
\end{equation*}
where $r(t)$, $t \in [a,b]$, is a parameterisation of a line with endpoints $r(a), r(b) \in \partial \mathcal{X}$. 
\ac{SED} for this problem was recently addressed, using bespoke code, in \cite{Burger_2021} and \cite{helin2021edgepromoting}. 
Following \citet{Burger_2021}, an experiment consists of a set of $9$ parallel line integrals across $\mathcal{X}$, with lines a perpendicular distance of $0.03$ apart.
As a toy example, we consider tomographic reconstruction of an indicator function $\mathsf{f}(x) = \mathbbm{1}_B(x)$ where $B$ is the ball of radius $0.3$ centred on $(0.4,0.4)$. 

For our statistical model $f$ we used a stationary \ac{GP} with Mat\'ern covariance and smoothness parameter $\nu = 2 + \frac{1}{2}$, and the non-linear quantity of interest was $q(\mathsf{f}) = \exp(3 \mathsf{f})$ which, when combined with squared error loss, serves to prioritise the reconstruction of the ball in \ac{SED}.
See \Cref{subsec: tomographic details} for full detail.

\begin{figure*}[t!]
   \centering
    \includegraphics[width=\textwidth]{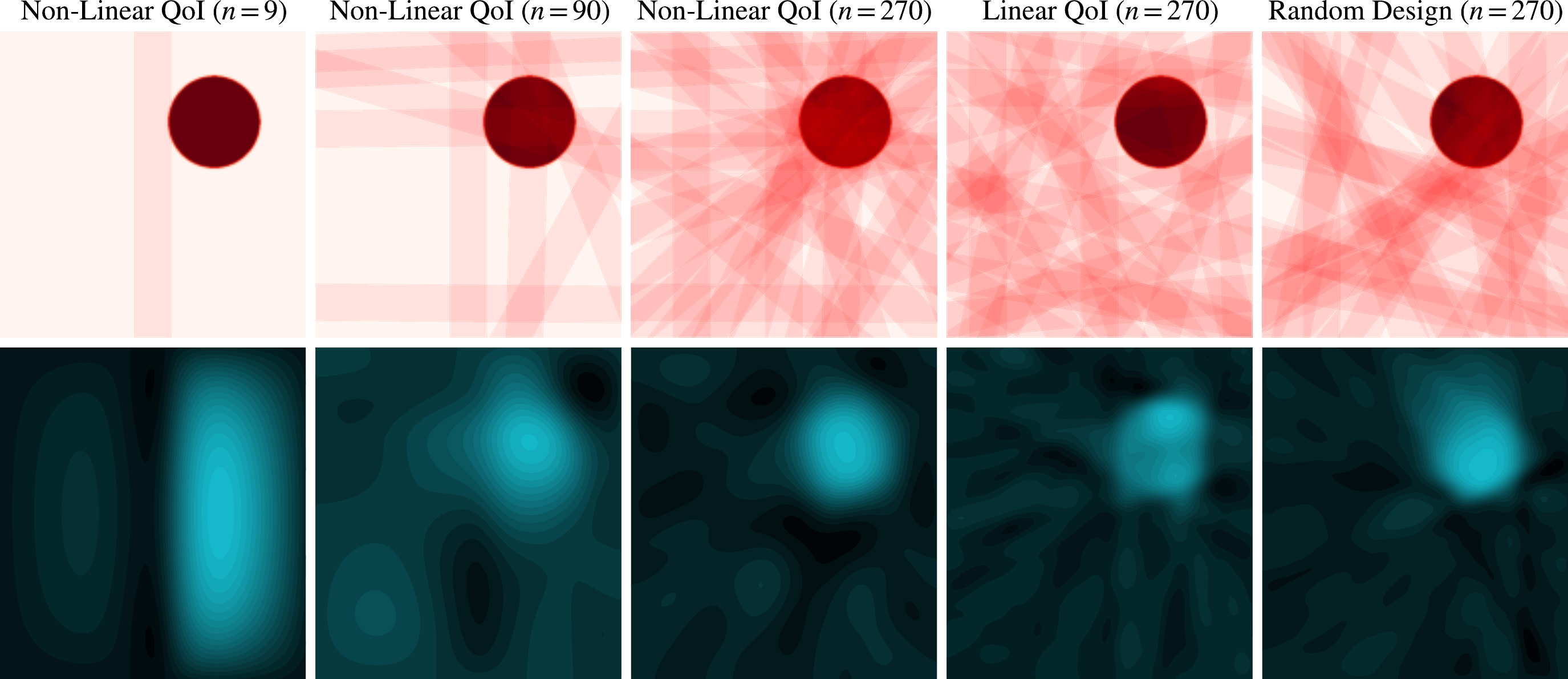} 
    
\caption{
    \textit{Tomographic Reconstruction:} The top row displays experimental designs, overlaying the latent $\mathsf{f}$. 
    Each red bar indicates the region over which $9$ equally-spaced line integrals were computed. The bottom row displays the corresponding mean of $f | \bm{\delta}_n(\mathsf{f})$, the posterior obtained using (from left to right): \ac{SED} with non-linear quantity of interest ($n = 9$, 90, 270), \ac{SED} with linear quantity of interest ($n=270$), and a random design ($n=270$).
}
\label{fig: tomographic reconstruction}
\end{figure*}

Results are shown in \Cref{fig: tomographic reconstruction} and only 32 lines of code were required.
In this experiment, we used $m = 28^2$ basis functions and began optimising hyperparameters at SED iteration $n = 1$. In total, 2.5 CPU hours were required.
\ac{SED} using \texttt{GaussED} provides improved reconstruction compared to a random design (right panel).
As an additional comparison, we also performed \ac{SED} with the linear quantity of interest $q(\mathsf{f}) = \mathsf{f}$ and a space-filling design was obtained.
Exploratory investigation of this kind is straight-forward in \texttt{GaussED}.

\subsection{Gradient-Based Bayesian Optimisation} \label{subsec: lotka volterra}

Our next example uses Bayesian optimisation to perform parameter inference via maximum likelihood, and for this we consider the Lotka--Volterra model
\begin{align}
\label{eq: LV}
\begin{split}
    \textstyle \frac{\mathrm{d}p}{\mathrm{d}t} &= \alpha p - \beta pq, \\ 
    \textstyle \frac{\mathrm{d}q}{\mathrm{d}t} &= - \gamma q + \delta pq , \\
\end{split}
\end{align}
where $p(t),q(t) > 0$ are the predator and prey populations, respectively, at time $t$ and $\alpha,\beta,\gamma$ and $\delta$ are free parameters to be inferred. 
To facilitate visualisation of experimental designs we consider inferring only $\alpha$ and $\beta$, 
which we collect in a single parameter vector $x = (\alpha,\beta)$. 
For this demonstration we restrict attention to $\mathcal{X} = [0.45,0.9] \times [0.09,0.5]$, to avoid failure of the numerical integrator applied to \eqref{eq: LV}.
The remaining parameters, $\gamma$ and $\delta$, are then taken as fixed. 
Our latent function $\mathsf{f}$ is the log-likelihood, denoted $\mathsf{f} = \log \mathcal{L}$, arising from a particular dataset of noise-corrupted observations described in \Cref{subsec: lotka-volterra details}.
Our quantity of interest is the maximum likelihood estimator $q(\mathsf{f}) = \max_{x \in\mathcal{X}} \mathsf{f}(x)$. 
The design set $\mathcal{D}$ contains pointwise evaluation functionals $\delta_x^1(\mathsf{f}) = \log \mathcal{L}(x)$ and gradient evaluation functionals $\delta_x^{2,i}(\mathsf{f}) = \nabla_{x_i} \log \mathcal{L}(x)$, and at each iteration of \ac{SED} we evaluate $(\delta_x^1(\mathsf{f}),\delta_x^{2,1}(\mathsf{f}),\delta_x^{2,2}(\mathsf{f}))$ for some $x \in \mathcal{X}$, mimicking the information provided when \eqref{eq: LV} is solved using an adjoint method.
Through a suitable sequence of evaluation functionals, \ac{SED} aims to approximate the maximum likelihood estimator.

\begin{figure*}[t!]
\centering
\begin{subfigure}[b]{0.225\linewidth}
\includegraphics[width=\textwidth]{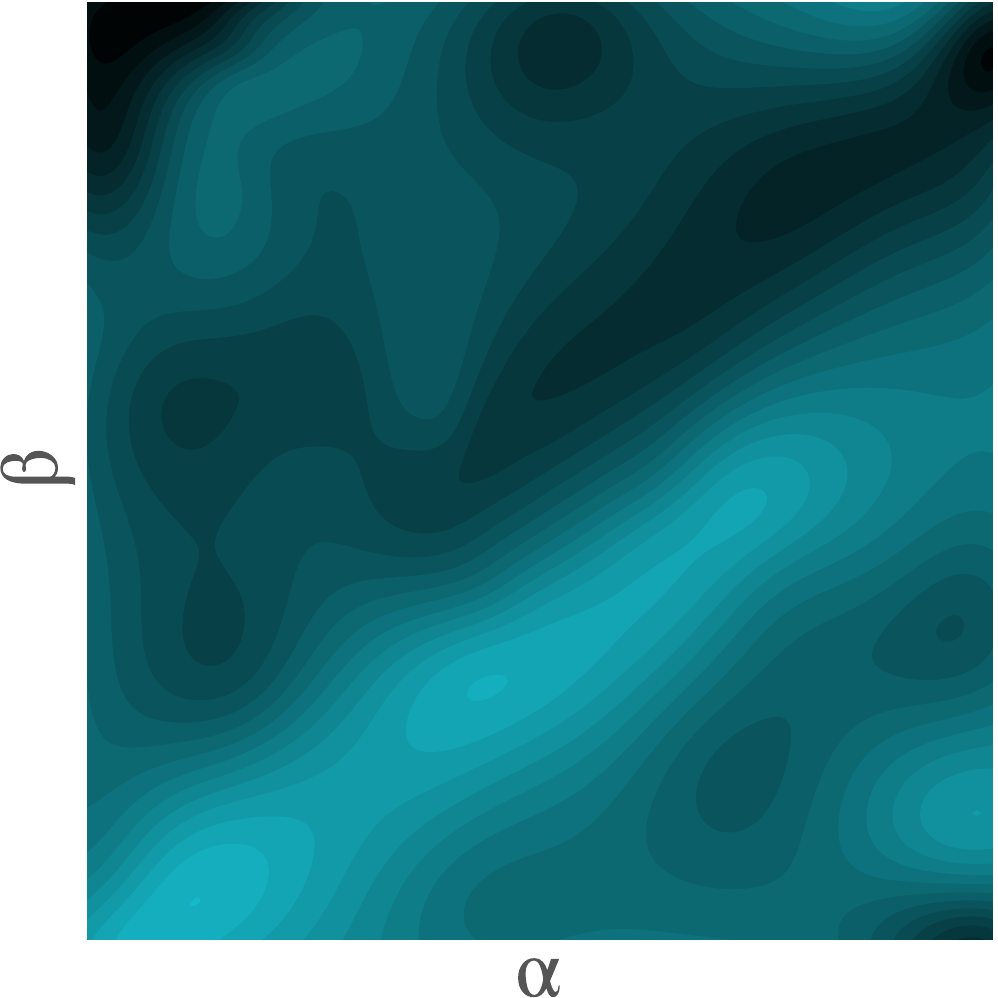} 
\caption{}
\end{subfigure}
\begin{subfigure}[b]{0.24\linewidth} 
\includegraphics[width=\textwidth]{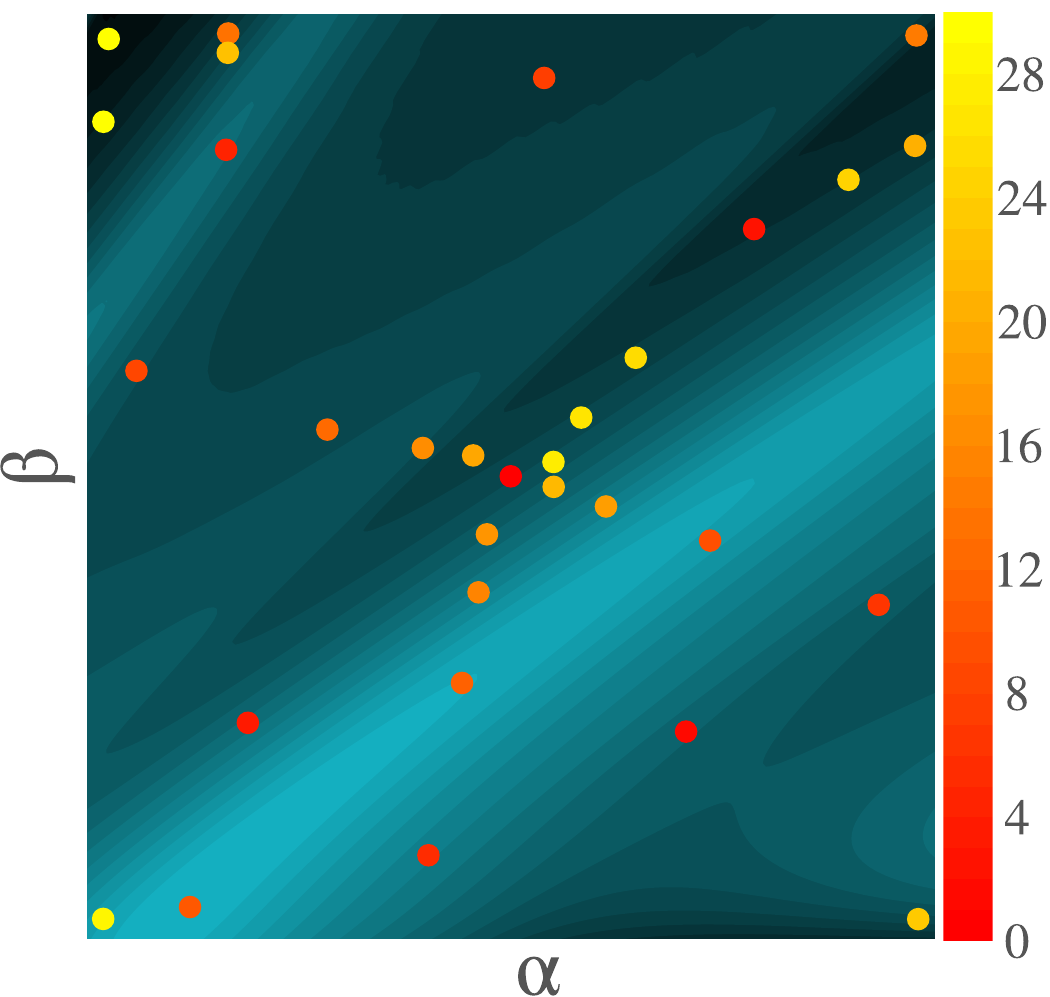}
\caption{} \label{fig: lv design bo}
\end{subfigure}
\begin{subfigure}[b]{0.24\linewidth}
\includegraphics[width=\textwidth]{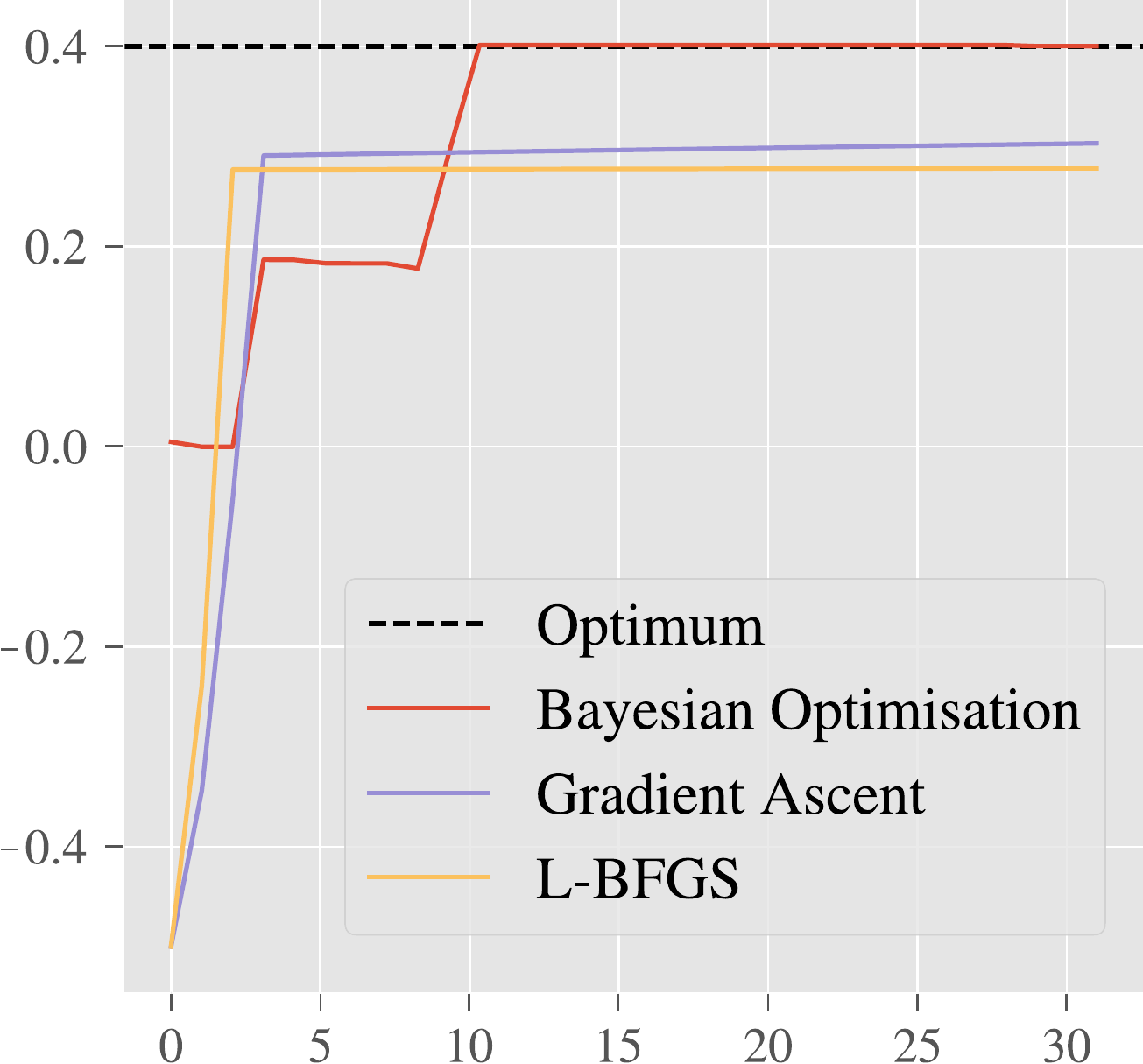} 
\caption{} \label{fig: lv convergence a}
\end{subfigure}
\begin{subfigure}[b]{0.24\linewidth}
\includegraphics[width=\textwidth]{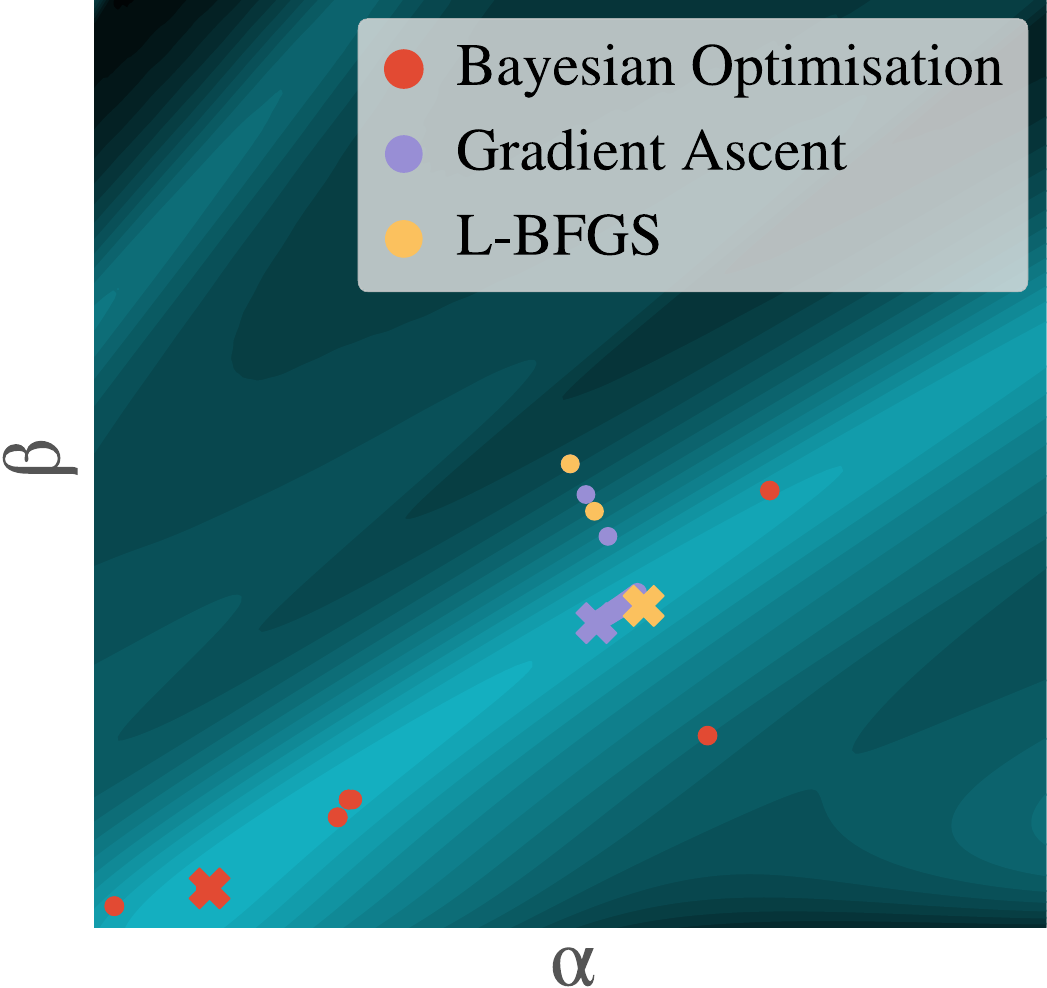}
\caption{} 
\end{subfigure}
\caption{ 
\textit{Gradient-Based Bayesian Optimisation:}
(a) Mean of $f | \bm{\delta}_n(\mathsf{f})$, the posterior after $n=90$ total evaluations.
(b) Log-likelihood $\mathsf{f}$, with design points overlaid. 
Colour indicates the order in which points were selected in \ac{SED}.
(c) Maximum value of the likelihood obtained during the first $m$ iterations of each optimisation method.
(d) Location of the maximum value along the optimisation path, where the colored \ding{54} symbols indicate the maximum value obtained (for Bayesian optimisation, the maximum of the posterior mean is reported).
}
\label{fig: lotka_volterra_error}   
\end{figure*}

Results are shown in \Cref{fig: lotka_volterra_error} and only 17 lines of code were required.
In this experiment we used $m=35^2$ basis functions, we computed $n_0 = 10$ iterations of SED before beginning hyperparameter optimisation and 1.5 CPU hours were required.
For reference, results based on gradient ascent and L-BFGS \citep{nocedal1980LBFGS} are also displayed. 
All algorithms were initialised at the midpoint of the domain $\mathcal{X}$ and run for $n=30$ iterations.
Bayesian optimisation with gradient data outperformed the first order optimisation methods in this example, where attention is focused on performance after a small number of likelihood evaluations, to mimic more challenging applications in which the likelihood is associated with a more substantial computational cost. 

\section{Discussion} \label{sec: discussion}

This paper introduced \texttt{GaussED}, a simple \ac{PPL} coupled to a powerful engine for \ac{SED}.
Through four experiments we illustrated the diverse applications that can be automatically solved using \texttt{GaussED}.
However, automation of \ac{SED} comes at a cost:
Firstly, \texttt{GaussED} is restricted both to continuous linear functional data and to \acp{GP}, limiting the potential for more flexible statistical models to be employed.
Alternative \acp{PPL}, such as \texttt{Emukit}, offer more modelling flexibility but require acquistion functions to be manually specified.
Secondly, in automating the specification of an acquisition function in \texttt{GaussED}, there may be a loss in performance terms compared to bespoke solutions for specific tasks.
Our experiments involving Bayesian optimisation in \Cref{subsec: lotka volterra} were encouraging, however, and suggested that such performance gaps, if they do exist, may be acceptably small.
One role for \texttt{GaussED} in these settings is to provide an off-the-shelf benchmark for \ac{SED}, against which more sophisticated methods can be compared.

\subsubsection*{Acknowledgements}
MAF was supported by the EPSRC Centre for Doctoral Training in Cloud Computing for Big Data EP/L015358/1 at Newcastle University, UK.
CJO was supported by the Lloyd's Register Foundation programme on data-centric engineering at the Alan Turing Institute, UK.
The authors thank Maren Mahsereci, Tim Sullivan and Darren Wilkinson for valuable insight.

\bibliographystyle{apalike}
\bibliography{References}

\newpage
\appendix
\newpage

\appendix

\onecolumn

\section*{Supplement}

The supplement is structured as follows:
\begin{itemize}
    \item \Cref{app: math prelim} contains the mathematical preliminaries for the subsequent sections \Cref{app: decision theory} and \Cref{app sec: properties of gaussian processes}.
    \item \Cref{app: decision theory} presents the conditions for the equivalence of \eqref{eq: U bayes risk} and \eqref{eq: rewrite U} as advertised in \Cref{subsec: decision theoretic approach}.
    \item \Cref{app sec: properties of gaussian processes} presents the formal background of conditioning on continuous linear data for Gaussian processes and presents properties of the Mat\'ern covariance function.
    \item \Cref{app: eigenfunctions} presents a derivation of the Gaussian process model that forms the foundation of \texttt{GaussED}. 
    \item \Cref{sec: gaussed details} discusses computational aspects of \texttt{GaussED}. In particular, we discuss linear algebra solvers, different approaches to sampling from the posterior, and we present a complete description of how \texttt{GaussED} attempts to optimise the acquisition function in SED.
    \item \Cref{sec: experiment details} contains full details of the experiments presented in \Cref{sec: demonstrations}. \Cref{subsec: heat equation details} details the \ac{PDE} experiment presented in \Cref{subsec: heat equation}. \Cref{subsec: tomographic details} details the tomographic reconstruction experiment presented in \Cref{subsec: tomographic reconstruction}. \Cref{subsec: lotka-volterra details} details the Lotka--Volterra experiment presented in \Cref{subsec: lotka volterra}. 
    \item \Cref{sec: evaluating Gaussed} presents further empirical evaluation of \texttt{GaussED}. \Cref{subsec: stochastic optimisation investigation} presents an empirical investigation of the stochastic optimisation methods used in SED. \Cref{subsec: basis function investigation} presents an empirical investigation on how the number of basis functions affects the quality of inference.
\end{itemize}

\section{Mathematical Preliminaries}
\label{app: math prelim}

In this section we present the mathematics required to ensure that the conditioning of stochastic processes in the main text is well-defined (\Cref{subsec: cond as disint}), as well as recalling the concept of a Fr\'{e}chet derivative (\Cref{subsec: Frechet}).

\subsection{Conditioning as Disintegration} \label{subsec: cond as disint}

In finite dimensions, conditioning of random variables can be performed using the density formulation of Bayes' theorem.
However, typical stochastic processes will be infinite-dimensional, meaning that (Lebesgue) densities do not exist in general.
This necessitates a level of mathematical abstraction to ensure that conditional probabilities are well-defined.
The appropriate notion, for this work, is that of \emph{disintegration}, defined next.

Let $(\mathcal{F},\mathcal{S}_{\mathcal{F}})$ and $(\mathcal{Y},\mathcal{S}_{\mathcal{Y}})$ be measurable spaces and let $\delta$ be a measurable function from $\mathcal{F}$ to $\mathcal{Y}$.
Recall that $\delta^{-1}(S) = \{f \in \mathcal{F} : \delta(f) \in S\}$ denotes the pre-image of $S \in \mathcal{S}_{\mathcal{Y}}$.
Let $\mathbb{P}$ be a probability measure on $(\mathcal{F},\mathcal{S}_{\mathcal{F}})$ and recall that $\delta_\# \mathbb{P}$ denotes the \emph{pushforward} measure $(\delta_\# \mathbb{P})(S) := \mathbb{P}(\delta^{-1}(S))$ on $\mathcal{Y}$.

\begin{definition} \label{def: disintegrate}
The collection $\{\mathbb{P}(\cdot | y ) \}_{y \in \mathcal{Y}}$ is called a $\delta$-\emph{disintegration} of $\mathbb{P}$ if
\begin{enumerate}
    \item $\mathbb{P}( \delta^{-1}(y) | y ) = 1$ for $\delta_\# \mathbb{P}$ almost all $y \in \mathcal{Y}$
\end{enumerate}
and, for each measurable function $g : \mathcal{F} \rightarrow [0,\infty)$, we have
\begin{enumerate}
    \setcounter{enumi}{1}
    \item $y \mapsto \int g(\mathrm{f}) \mathrm{d}\mathbb{P}(\mathrm{f} | y)$ is measurable
    \item $\int g(\mathrm{f}) \mathrm{d}\mathbb{P}(\mathrm{f}) = \int \int g(\mathrm{f}) \mathrm{d}\mathbb{P}(\mathrm{f} | y) \mathrm{d} \delta_\# \mathbb{P}(y)$
\end{enumerate}
\end{definition}

A disintegration is a particular instance of a \emph{regular conditional distribution} which also satisfies property (1) in \Cref{def: disintegrate}; see \cite{chang1997conditioning}.
A basic theorem on the existence and $\delta_\# \mathbb{P}$ almost everywhere uniqueness of disintegrations is given in \citet[][p147]{parthasarathy2005probability}.
Two disintegrations will be identified if they coincide $\delta_\# \mathbb{P}$ almost everywhere, and we will therefore refer to \emph{the} $\delta$-disintegration of $\mathbb{P}$.
The concept of disintegration makes precise what it means to ``condition \acp{GP} on data'', as discussed in \Cref{subsec: disintegration of gaussian}.

\subsection{Fr\'{e}chet Derivatives} \label{subsec: Frechet}

Recall that $\mathcal{F}$ was defined as a normed vector space, meaning that the notion of a \emph{Fr\'{e}chet derivative} can be exploited.
A function $q : \mathcal{F} \rightarrow \mathbb{R}^d$ is called \emph{Fr\'{e}chet differentiable} at $\mathrm{f} \in \mathcal{F}$ if there exists a bounded linear operator $A : \mathcal{F} \rightarrow \mathbb{R}^d$ such that
$$
\lim_{\|\mathrm{g}\| \rightarrow 0} \frac{\|q(\mathrm{f} + \mathrm{g}) - q(\mathrm{f}) - A(\mathrm{g})\|}{\|\mathrm{g}\|} = 0 .
$$
If such an operator exists it can be shown to be unique, called the \emph{Fr\'{e}chet derivative} of $q$ at $\mathrm{f}$, and denoted $\mathrm{D}q(\mathrm{f}) = A$.
To emphasise that the Fr\'{e}chet derivative is an operator, we occasionally write $\mathrm{D}q(\mathrm{f})(\cdot)$ in the sequel.
A Fr\'{e}chet derivative $\mathrm{D}q(f)$ is said to have \emph{full rank} if $\mathrm{D}q(\mathrm{f})(\mathrm{g}) = 0$ implies $\mathrm{g} = 0$.

The chain rule for Fr\'{e}chet derivatives takes the form
$$
\mathrm{D}(b \circ a)(\mathrm{f})(\cdot) = (\mathrm{D} b \circ a)(\mathrm{f}) \circ \mathrm{D} a (\mathrm{f})(\cdot) .
$$
As a concrete example, that we use later, consider $a(\mathrm{f}) = q(\mathrm{f})$ to be the quantity of interest and $b(x) = \|x - q(\mathrm{g})\|^2$ for all $x \in \mathbb{R}^d$ and some fixed $g \in \mathcal{F}$.
Then $\mathrm{D}b(x)(\cdot) = 2 \langle x - q(\mathrm{g}) , \cdot \rangle$ is a linear operator from $\mathbb{R}^d$ to $\mathbb{R}$ and we have
\begin{align}
    \mathrm{D}(b \circ a)(\mathrm{f})(\cdot) = 2 \langle q(\mathrm{f}) - q(\mathrm{g}) , \mathrm{D} q(\mathrm{f})(\cdot) \rangle , \label{eq: Fre chain}
\end{align}
which is a linear operator from $\mathcal{F}$ to $\mathbb{R}$.
Further background on Fr\'{e}chet derivatives can be found in \citet[][Section 2.1C]{berger1977nonlinearity}.

An important technical result on Fr\'{e}chet derivatives, that we will use in the sequel, is when the interchange of a Fr\'{e}chet derivative and an integral can be permitted:

\begin{proposition} \label{prop: change order}
Let $\mathcal{F}$ be complete (i.e. a Banach space) and $(\Omega,\mathcal{S},\mathbb{P})$ be a probability space.
Let $\ell : \mathcal{F} \times \Omega \rightarrow \mathbb{R}$ satisfy the following:
\begin{enumerate}
    \item $\mathrm{f} \mapsto \ell(\mathrm{f},\omega)$ is Fr\'{e}chet differentiable, for each $\omega \in \Omega$
    \item $\omega \mapsto \ell(\mathrm{f},\omega)$ is integrable, for each $\mathrm{f} \in \mathcal{F}$
    \item $\omega \mapsto \mathrm{D}\ell(\mathrm{f},\omega)(\mathrm{g})$ is integrable, for each $\mathrm{f}, \mathrm{g} \in \mathcal{F}$
    \item $\int \|\mathrm{D} \ell(\mathrm{f},\omega) \| \mathrm{d}\mathbb{P}(\omega) < \infty$
\end{enumerate}
Then the function
$$
r(\mathrm{f}) := \int \ell(\mathrm{f},\omega) \mathrm{d}\mathbb{P}(\omega)
$$
is Fr\'{e}chet differentiable, with derivative
$$
\mathrm{D}r(\mathrm{f})(\cdot) = \int \mathrm{D}\ell(\mathrm{f},\omega)(\cdot) \mathrm{d}\mathbb{P}(\omega).
$$
\end{proposition}
\begin{proof}
A special case of \citet{Frechet}.
\end{proof}

\section{Regularity Conditions for the Decision Theoretic Formulation} \label{app: decision theory}

The aim in this section is to establish sufficient conditions for the equivalence of \eqref{eq: U bayes risk} and \eqref{eq: rewrite U} as advertised in \Cref{subsec: decision theoretic approach}.
To achieve this, we will use the notion of a Fr\'{e}chet derivative from in \Cref{subsec: Frechet}.
Our sufficient conditions are presented in \Cref{subsec: suffic cond}.
A short discussion of the strength of these conditions is contained in \Cref{subsec: discuss conds}.

\subsection{From Optimisation to Expectation} \label{subsec: suffic cond}

Firstly, we rigorously establish an infinite-dimensional analogue of the classical result that the posterior mean is a Bayes act for squared error loss:

\begin{proposition} \label{prop: Bayes act}
Let $L(\mathrm{f},\mathrm{g}) = \|q(\mathrm{f}) - q(\mathrm{g})\|^2$.
Assume that $\mathcal{F}$ is complete (i.e. a Banach space) and that:
\begin{enumerate}
    \item[{\normalfont (A1)}]  $q : \mathcal{F} \rightarrow \mathbb{R}^d$ is Fr\'{e}chet differentiable;
    \item[{\normalfont (A2)}] the Fr\'{e}chet derivative $\mathrm{D}q(\mathrm{f})$ has full rank at all $\mathrm{f} \in \mathcal{F}$;
    \item[{\normalfont (A3)}] $\int \|q(\mathrm{g})\|^2 \mathrm{d}\mathbb{P}_f(\mathrm{g} | \bm{\delta}_{n}(\mathsf{f})) < \infty$.
\end{enumerate}
Then any solution to
\begin{align}
    \argmin_{\mathrm{f} \in \mathcal{F}} r(\mathrm{f}), \qquad r(\mathrm{f}) := \int L(\mathrm{f},\mathrm{g}) \; \mathrm{d}\mathbb{P}_f(\mathrm{g} | \bm{\delta}_{n}(\mathsf{f}) ) \label{eq: Bayes act}
\end{align}
satisfies
\begin{align*}
    q(\mathrm{f}) = \int q(\mathrm{g}) \; \mathrm{d}\mathbb{P}_f(\mathrm{g} | \bm{\delta}_{n}(\mathsf{f}) ) .
\end{align*}
\end{proposition}
\begin{proof}
From an application of \Cref{prop: change order} with $\ell(\mathrm{f},\omega) = L(\mathrm{f},g(\omega))$, where $g : \Omega \rightarrow \mathcal{F}$ is a random variable with distribution $\mathbb{P}_f(\cdot | \bm{\delta}_n(\mathsf{f}))$, we deduce that our assumptions on $L$ and $q$ (A1) are sufficient for the Fr\'{e}chet derivative of $r$ to exist.
Thus, a minimiser $\mathrm{f}$ of \eqref{eq: Bayes act} satisfies $\mathrm{D}r(\mathrm{f}) = 0$.
To evaluate $\mathrm{D}r$ we exploit the integrability assumption (A3) on $q$ to differentiate under the integral, which is also justified from \Cref{prop: change order}:
\begin{align*}
    \mathrm{D}r(\mathrm{f})(\cdot) = \int \mathrm{D}L(\mathrm{f},\mathrm{g}) \; \mathbb{P}_f(\mathrm{g} | \bm{\delta}_n(\mathsf{f})) .
\end{align*}
Next we apply the chain rule for Fr\'{e}chet derivatives in the form of \eqref{eq: Fre chain}, yielding
\begin{align*}
    \mathrm{D}r(\mathrm{f})(\cdot) & = \int 2 \langle q(\mathrm{f}) - q(\mathrm{g}) , \mathrm{D}q(\mathrm{f})(\cdot) \rangle \; \mathrm{d}\mathbb{P}_f(\mathrm{g} | \bm{\delta}_{n}(\mathsf{f}) ) \\
    & = 2 \left\langle \underbrace{ q(\mathrm{f}) - \int q(\mathrm{g}) \mathrm{d}\mathbb{P}_f(\mathrm{g} | \bm{\delta}_{n}(\mathsf{f}) ) }_{(*)} , \mathrm{D}q(\mathrm{f})(\cdot) \right\rangle .
\end{align*}
Since $\mathrm{D}q(\mathrm{f})$ was assumed to have full rank (A2), if $\mathrm{D}r(\mathrm{f}) = 0$ then $(*) = 0$, whence the claimed result.
\end{proof}

Now we are able to prove the advertised result:

\begin{proposition}
In the setting of \Cref{prop: Bayes act}, and under assumptions {\normalfont (A1-3)}, we have
\begin{align*}
    \min_{\mathrm{f} \in \mathcal{F}} r(\mathrm{f}) = \frac{1}{2} \iint  L(\mathrm{g},\mathrm{g}') \; \mathrm{d}\mathbb{P}_f(\mathrm{g} | \bm{\delta}_{n}(\mathsf{f})   ) \; \mathrm{d}\mathbb{P}_f(\mathrm{g}' | \bm{\delta}_{n}(\mathsf{f})  ) .
\end{align*}
\end{proposition}
\begin{proof}
Let $\mathrm{f} \in \mathcal{F}$ solve \eqref{eq: Bayes act}.
Then consider the algebraic identity
\begin{align*}
    q(\mathrm{g}) - q(\mathrm{g}') = \left\{ q(\mathrm{g}) - q(\mathrm{f}) \right\} - \left\{ q(\mathrm{g}') - q(\mathrm{f}) \right\} .
\end{align*}
Using this identity, the loss function can be expressed as
\begin{align*}
    L(\mathrm{g},\mathrm{g}') & = \|q(\mathrm{g}) - q(\mathrm{g}')\|^2 \\
    & = \| q(\mathrm{g}) - q(\mathrm{f}) \|^2 - 2 \langle q(\mathrm{g}) - q(\mathrm{f}) , q(\mathrm{g}') - q(\mathrm{f})  \rangle + \| q(\mathrm{g}') - q(\mathrm{f}) \|^2 .
\end{align*}
From linearity of the inner product we have that
\begin{align*}
    & \iint \langle q(\mathrm{g}) - q(\mathrm{f}) , q(\mathrm{g}') - q(\mathrm{f})  \rangle \; \mathrm{d}\mathbb{P}_f(\mathrm{g} | \bm{\delta}_{n}(\mathsf{f})  ) \; \mathrm{d}\mathbb{P}_f(\mathrm{g}' | \bm{\delta}_{n}(\mathsf{f})  )  \\
    & = \left\langle \underbrace{ \int q(\mathrm{g}) - q(\mathrm{f}) \; \mathrm{d}\mathbb{P}_f(\mathrm{g} | \bm{\delta}_{n}(\mathsf{f})  ) }_{=0} , \underbrace{ \int q(\mathrm{g}') - q(\mathrm{f}) \; \mathrm{d}\mathbb{P}_f(\mathrm{g}' | \bm{\delta}_{n}(\mathsf{f})  ) }_{=0} \right\rangle = 0 ,
\end{align*}
where we have used the integrability assumption (A3) on $q$ to bring the integrals into the inner product, and we have used \Cref{prop: Bayes act} to conclude that each argument is equal to 0.
Finally, from the fact that $g$ and $g'$ are identically distributed, we have
\begin{align*}
    & \frac{1}{2} \iint  L(\mathrm{g},\mathrm{g}') \; \mathrm{d}\mathbb{P}_f(\mathrm{g} | \bm{\delta}_{n}(\mathsf{f})   ) \; \mathrm{d}\mathbb{P}_f(\mathrm{g}' | \bm{\delta}_{n}(\mathsf{f})  ) \\
    & = \frac{1}{2} \iint \|q(\mathrm{g}) - q(\mathrm{f})\|^2 + \|q(\mathrm{g}') - q(\mathrm{f})\|^2 \; \mathrm{d}\mathbb{P}_f(\mathrm{g} | \bm{\delta}_{n}(\mathsf{f})   ) \; \mathrm{d}\mathbb{P}_f(\mathrm{g}' | \bm{\delta}_{n}(\mathsf{f})  ) \\
    & = \frac{1}{2} \times 2 \times \int \|q(\mathrm{g}) - q(\mathrm{f})\|^2 \; \mathrm{d}\mathbb{P}_f(\mathrm{g} | \bm{\delta}_{n}(\mathsf{f})   ) = \min_{\mathrm{f} \in \mathcal{F}} r(\mathrm{f}) ,
\end{align*}
which completes the argument.
\end{proof}

\subsection{Verifying the Assumptions} \label{subsec: discuss conds}

The main assumption in \Cref{prop: Bayes act} is (A2); the requirement that $\mathrm{D}q(\mathrm{f})$ has full rank for all $\mathrm{f} \in \mathcal{F}$.
As we explain below through a worked example, (A2) is non-trivial but may often be satisfied with only minor modification to the \ac{SED} task in hand.

As a worked example, suppose $\mathcal{F}$ is a Hilbert space containing smooth, real-valued functions defined on a compact set $\mathcal{X} \subset \mathbb{R}^d$.
Suppose that we are interested in the quantity of interest
\begin{align}
q(\mathrm{f}) = \int_{\mathcal{X}} \mathrm{f}(x) \mathrm{d}x . \label{eq: q1}
\end{align}
Then (A2) is \emph{not} satisfied in general, because $q(\mathrm{f} + \mathrm{g}) =  q(\mathrm{f})$ for all $\mathrm{g}$ in the linear subspace $\mathcal{G} = \{ \mathrm{g} \in \mathcal{F} : \int_{\mathcal{X}} \mathrm{g}(x) \mathrm{d}x = 0 \}$.
It follows that $\mathrm{D}q(\mathrm{f})(\mathrm{g}) = 0$ for all $\mathrm{g} \in \mathcal{G}$, so that $\mathrm{D}q(\mathrm{f})$ does not have full rank whenever $\mathcal{G}$ is non-trivial.
However, (A2) \emph{is} satisfied if we restrict attention to the normed vector space $\mathcal{F}_{\textsf{c}}$ spanned by the elements of $\mathcal{F} \setminus \mathcal{G}$, since then $\mathrm{D}q(\mathrm{f})(\mathrm{g}) = \int_{\mathcal{X}} \mathrm{g}(x) \mathrm{d}x$ and thus $\mathrm{D}q(\mathrm{f})(\mathrm{g}) = 0$ with $g \in \mathcal{F}_{\textsf{c}}$ implies $g = 0$.
This illustrates that, with a small amount of technical care, the assumptions of \Cref{prop: Bayes act} can often be satisfied.



\section{Properties of Gaussian Processes} \label{app sec: properties of gaussian processes}

In this section, we present the formal background of conditioning on continuous linear data for Gaussian processes and detail properties of the Mat\'ern covariance function.

\subsection{Disintegration of Gaussian Measures} \label{subsec: disintegration of gaussian}

Let $\mathcal{X}$ be a compact subset of $\mathbb{R}^d$ for some $d \in \mathbb{N}$ and let $C^r(\mathcal{X})$ denote the vector space of $r$-times continuously differentiable real-valued functions on $\mathcal{X}$ equipped with the norm 
$$
\|\mathrm{f}\|_{C^r(\mathcal{X})} = \max_{|\alpha| \leq r} \|\mathrm{f}^{(\alpha)}\|_\infty ,
$$
where the maximum ranges over multi-indices $\alpha \in \mathbb{N}_0^d$ with $|\alpha| = \alpha_1 + \dots + \alpha_d \leq r$ and $\mathrm{f}^{(\alpha)}(x) := \partial_{x_1}^{\alpha_1} \dots \partial_{x_d}^{\alpha_d} \mathrm{f}(x)$.
In what follows we consider disintegration in the case where $\mathcal{F} = C^r(\mathcal{X})$, equipped with the Borel $\sigma$-algebra, and $\mathcal{Y} = \mathbb{R}$.
For an operator $\delta$ and a bivariate function $k(\cdot,\cdot)$, denote $\delta k(\cdot,\cdot)$ to be the action of $\delta$ on the first argument of $k$, and denote $\bar{\delta} k(\cdot,\cdot)$ to be the action of $\delta$ on the second argument of $k$.

\begin{lemma} \label{lem: Gauss disin}
Let $\mathbb{P}$ be a Gaussian measure on $C^r(\mathcal{X})$ with mean function $m : \mathcal{X} \rightarrow \mathbb{R}$ and covariance function $k : \mathcal{X} \times \mathcal{X} \rightarrow \mathbb{R}$.
Let $\delta : C^r(\mathcal{X}) \rightarrow \mathbb{R}$ be a continuous linear functional.
For each $y \in \mathbb{R}$, define $\mathbb{P}(\cdot | y)$ to be a Gaussian measure with mean and covariance function
\begin{align*}
    m_y(x) & = m(x) + [\bar{\delta} k(x,\cdot)] [\delta \bar{\delta} k(\cdot,\cdot)]^{-1} (y - m(x)) \\
    k_y(x,x') & = k(x,x') - [\bar{\delta} k(x,\cdot)] [\delta \bar{\delta} k(\cdot,\cdot)]^{-1} [\delta k(\cdot,x')] .
\end{align*}
Then $\{\mathbb{P}(\cdot | y)\}_{y \in \mathbb{R}}$ is a $\delta$-disintegration of $\mathbb{P}$.
\end{lemma}
\begin{proof}
The proof is by direct verification of properties (1-3) in \Cref{def: disintegrate}.
See e.g. p.188 of \cite{ritter2007average}.
\end{proof}

The fact that the elements $\mathbb{P}(\cdot|y)$ of the disintegration are again Gaussian enables the repeated application of \Cref{lem: Gauss disin}, for example to condition on $n \geq 1$ continuous linear functionals $\bm{\delta}_n = (\delta_1, \dots, \delta_n)^\top$, as exploited in the main text.
Constructed in this way, it can be verified that the elements $\mathbb{P}(\cdot | \bm{y}_n)$ of the resulting disintegration, with $\bm{y}_n \in \mathcal{Y}^n$, are invariant to the order in which the disintegrations are performed.

\subsection{Sufficient Conditions for Disintegration of Mat\'{e}rn Processes} \label{app subsec: sufficient matern}

To exploit \Cref{lem: Gauss disin} in practice it is sufficient to verify that samples from the Gaussian measure are almost surely contained in $C^r(\mathcal{X})$.
Such analysis is technical but specific results are available for derivative data in the context of the tensor product Mat\'{e}rn covariance model that we primarily use in this work.
Indeed, let $\mathbb{P}$ be a Gaussian measure with mean function $m$ and covariance function $k$, such that $m \in C^r(\mathcal{X})$ and 
\begin{align}
    k(x,x') & := \sigma^2 \prod_{i=1}^d k_{\nu_i}(x_i - x_i') \label{eq: tensor matern} , \qquad 
    k_{\nu}(z)  := \frac{2^{1-\nu}}{\Gamma(\nu)} \left( \sqrt{2 \nu} \frac{|z|}{\rho} \right)^\nu K_\nu\left( \sqrt{2 \nu} \frac{|z|}{\rho} \right) , 
\end{align}
where $K_\nu$ denotes the modified Bessel function of the second kind and $\nu_i := r + \frac{1}{2}$.
Then Theorem 2 of \cite{wang2021bayesian} establishes that samples are almost surely contained in $C^r(\mathcal{X})$.
Moreover, maps of the form $\delta(f) = f^{(\alpha)}(x)$, $|\alpha| \leq r$ are continuous linear functionals from $C^r(\mathcal{X})$ to $\mathbb{R}$, since $|\delta(f)| \leq \|f\|_{C^r(\mathcal{X})}$ for all $f \in C^r(\mathcal{X})$.
Thus, in this case \Cref{lem: Gauss disin} can be used to condition the tensor product Mat\'{e}rn process in \eqref{eq: tensor matern} on the derivative data $f^{(\alpha)}(x)$, safe in the knowledge that the conditional process will be well-defined.

\section{Spectral Approximation}
\label{app: eigenfunctions}

This section presents an informal derivation of the spectral \ac{GP} approximation of \cite{solin2019hilbert}. 
The following utilises properties of the Fourier transform, which are first briefly recalled.

\subsection{Properties of the Fourier Transform} 

In the following we use $F$ to denote the Fourier transform operator and use the notation $\hat{f} \coloneqq F(f)$ to denote the Fourier transform of $f$. In the following we use the convention of using  the angular frequency domain. Therefore, for square-integrable $f:\mathbb{R}^d\rightarrow\mathbb{R}$, we have
\begin{equation*}
    F(f) = \frac{1}{(2\pi)^d} \int f(x) \exp(i\langle \omega, x\rangle)\,\mathrm{d}x.
\end{equation*}
Recall that, when an operator $T$ satisfies $F(Tf)(\omega) = m(\omega) \hat{f}(\omega)$, the operator $T$ is called a \emph{multiplier operator} and the corresponding $m$ is called the \emph{multiplier} of $T$. As a trivial example, the identity operator $Tf = f$ is a multiplication operator, with associated multiplier $1$. A more elaborate example, that is used in the subsequent section, is the Laplace operator $\Delta \coloneqq \frac{\partial^2}{\partial x_1^2} + \ldots + \frac{\partial^2}{\partial x_d^2}$, acting on twice differentiable functions $f:\mathbb{R}^d\rightarrow\mathbb{R}$. It can be shown that
\begin{equation} \label{eq: laplacian multiplier}
    F(\Delta f) = - \|\omega\|^2 F(f).
\end{equation}
Therefore, the Laplace operator is a multiplier operator with corresponding multiplier $-\|\omega\|^2$. Similarly, compositions of Laplace operators $\Delta^n \coloneqq \underbrace{\Delta \circ \ldots \circ \Delta}_{n \text{ times}}$, acting on sufficiently smooth functions $f$, is also a multiplier operator with multiplier $(-\|\omega\|^2)^n$. This can be seen by induction on the previous formula,
\begin{equation*}
    F(\Delta^n f) = - \|\omega\|^2 F(\Delta^{n-1}f)= \ldots = (- \|\omega\|^2)^n F(f).
\end{equation*}
By the convolution theorem, every multiplier operator $T$ with multiplier $m_T$, has an associated convolution kernel $k_T \coloneqq F^{-1}(m_T)$ that satisfies the following
\begin{align*}
F(Tf)(\omega) &= m_T(\omega) \hat{f}(\omega) \\
    Tf &= F^{-1}(m_T \hat{f}) 
    = f \star F^{-1}(m_T) 
    = f \star k_T,
\end{align*}
where $\star$ denotes convolution. Thus a multiplier operator is, in this sense, equivalent to a convolution operation.

We now state two important results that define the intimate connection between covariance functions and the Fourier transform. The first result is known as \emph{Bochner's theorem} \citep{rudin1990fourier}.

\begin{theorem}[Bochner's theorem] \label{theorem: bochner's theorem}
A \emph{stationary} covariance function, i.e. a covariance function of the form $k(x,y) = k(x-y)$, $k:\mathbb{R}^d\rightarrow\mathbb{R}$, can be written as the inverse Fourier transform of a finite positive measure $\mu$ such that $k(0) = \mu(\mathbb{R}^d)$. That is
\begin{equation*}
    k(x) = \frac{1}{(2\pi)^d}\int \exp\left(i \langle \omega, x\rangle \right) \,\mathrm{d}\mu(\omega).
\end{equation*}
\end{theorem}
The measure $\mu$ is called the \emph{spectral measure} of $k$ and the density of $\mu$, if it exists, is called the \emph{spectral density} $s(\omega)$ of $k$. In the case where the spectral density $s(\omega)$ of a stationary covariance function $k$ exists, $k$ and $s$ exist as Fourier duals. This result is known as the \emph{Wiener--Khintchine theorem} \citep{Khintchine1934theorem}.

\begin{theorem}[Wiener--Khintchine theorem] \label{theorem: wiener khintchin}
Suppose that the spectral density $s:\mathbb{R}^d\rightarrow\mathbb{R}$ of a stationary covariance function $k:\mathbb{R}^d\rightarrow\mathbb{R}$ exists, then
\begin{align*}
    k(x) &= \frac{1}{(2\pi)^d}\int s(\omega) \exp\left(i \langle \omega, x\rangle \right) \, \mathrm{d}\omega, \qquad
    s(\omega) = \int k(x) \exp\left(-i \langle \omega, x\rangle \right) \, \mathrm{d}s.
\end{align*}
\end{theorem}

In the proceeding section the Wiener--Khintchine theorem and the equivalence between a multiplier operator and an associated convolution operation are both used to establish a correspondence between the covariance operator of a stationary kernel $k$ and its spectral density $s$. This is the foundation upon which the spectral GP approximation of \cite{solin2019hilbert} is established.

\subsection{Spectral Gaussian Processes}

For every covariance function $k$, there exists an associated Hilbert--Schmidt integral operator, termed the \emph{covariance operator},
\begin{equation*}
    \mathcal{K} f = \int k(\cdot, y) f(y)\,\mathrm{d}y.
\end{equation*}
When $k$ is stationary, the resulting covariance operator takes the form of a convolution
\begin{equation*}
    \mathcal{K} f(x) = \int k(x - y)f(y)\,\mathrm{d}y = (f \star k)(x).
\end{equation*}
By the convolution theorem, we can then write the operator in the form $F(\mathcal{K}f) = \hat{k} \hat{f}$ and so $\mathcal{K}$ is a multiplier operator with multiplier $\hat{k}$. By \Cref{theorem: wiener khintchin}, the multiplier of $\mathcal{K}$ is the spectral density $s = \hat{k}$ of $k$. 

Assuming now that the covariance function is isotropic and so satisfies 
\begin{equation*}
    k(x,y) = k(\|x-y\|),
\end{equation*}
the corresponding spectral density $s$ of $k$ can be written as a function of $\|\omega\|$ only and so $s(\omega) = S(\|\omega\|)$, for an appropriate function $S$. As a further manipulation, we can write $s$ as a function of $\|\omega\|^2$ only, $s(\omega) = \psi(\|\omega\|^2)$. Assuming that $\psi$ possesses a Taylor expansion, we can write
\begin{equation*}
    s(\omega) = \psi(\|\omega\|^2) = \sum_{i=0}^\infty \mu_i (\|\omega\|^2)^i,
\end{equation*}
with each $\mu_i \in\mathbb{R}$. Inspired by the multiplier $-\|\omega\|^2$ of the Laplacian in \eqref{eq: laplacian multiplier} and by utilising the above Taylor expansion, we can write the Fourier transform of the covariance operator of an isotropic kernel in the form
\begin{align*}
    F(\mathcal{K} f)(\omega) &= s (\omega) \hat{f}(\omega) 
                     =  \sum_{i=0}^\infty \mu_i (\|\omega\|^2)^i\hat{f}(\omega) 
                     =  \sum_{i=0}^\infty \mu_i F( (-\Delta)^i f).
\end{align*}
By continuity of $F$, taking the inverse Fourier transform of the above yields a polynomial expansion form of the covariance operator
\begin{equation} \label{eq: hilbert schmidt laplacian}
    \mathcal{K}f = \sum_{i=0}^\infty \mu_i (-\Delta)^if.
\end{equation}

The remaining step is to approximate the negative Laplacian operator. To achieve this, we write the convolution kernel $k_{-\Delta}$ of the negative Laplacian as a Mercer expansion. To this end, we consider the following eigenvalue problem of the Laplacian over a compact domain $\mathcal{X} \subseteq \mathbb{R}^d$, with boundary $\partial \mathcal{X}$, with Dirichlet boundary conditions
\begin{alignat}{2} \label{eq: eigenvalue problem}
    -\Delta \phi_i(x) & = \lambda_i \phi_i(x), & \quad & x \in \mathcal{X}, \\
    \phi_i(x) & = 0, & \quad & x \in \partial \mathcal{X}. \label{eq: boundary conditions}
\end{alignat}
Over a suitable domain contained within $L^2(\mathcal{X})$, the negative Laplacian is a positive definite Hermitian operator and so we can provide a Mercer expansion of the convolution kernel $k_{-\Delta}$ of the negative Laplacian, utilising the eigenfunctions $\phi_i$. Similarly, we can provide a Mercer expansion of the convolution kernel of $(-\Delta)^n$, noting that each $\phi_i$ is again an eigenfunction, but now with corresponding eigenvalue $\lambda_i^n$. This can be seen by iteratively applying $-\Delta$ to the eigenvalue problem \eqref{eq: eigenvalue problem}. Therefore, we have
\begin{align*}
    (-\Delta)^n f(x) &= f \star k_{(-\Delta)^n} (x) 
                     =\int k_{(-\Delta)^n}(x-y) f(y) \,\mathrm{dy},
\end{align*}
where 
\begin{equation*}
    k_{(-\Delta)^n}(x-y) = \sum_{j=1}^\infty \lambda_j^n \phi_j(x)\phi_j(y).
\end{equation*}
Plugging the preceding formula into equation \eqref{eq: hilbert schmidt laplacian} yields the following:
\begin{align*}
    \mathcal{K}f(x) &= \sum_{i=0}^\infty \mu_i (-\Delta)^if 
    = \sum_{i=0}^\infty \mu_i \int k_{(-\Delta)^i}(x-y) f(y) \,\mathrm{dy} 
    =  \int \left(\sum_{i=0}^\infty \mu_i k_{(-\Delta)^i}(x-y) \right) f(y) \,\mathrm{dy}.
\end{align*}
Comparing the above form of $\mathcal{K}f(x)$ to its original definition $\mathcal{K}f(x) = \int k(x-y) f(y) \,\mathrm{d}y$ implies that we can approximate $k$ as follows
\begin{align*}
    k(x,y) &\approx \sum_{i=0}^\infty \mu_i k_{(-\Delta)^i}(x-y) 
    = \sum_{i=0}^\infty \mu_i \sum_{j=1}^\infty \lambda_j^i \phi_j(x)\phi_j(y) 
    = \sum_{j=1}^\infty \left(\sum_{i=0}^\infty \mu_i \lambda_j^i \right) \phi_j(x)\phi_j(y) 
    = \sum_{j=1}^\infty s(\sqrt{\lambda_j}) \phi_j(x)\phi_j(y),
\end{align*}
where, in the final step, we utilised our Taylor expansion of the spectral density $s$ of $k$ and set $\|\omega\|^2 = \lambda_j$ for each $j \in \mathbb{N}$. Refer to the original work \cite{solin2019hilbert} for convergence analyses of the given approximation.

Therefore, the resulting Gaussian model assumes the following truncated basis expansion
\begin{equation*} \label{eq: truncated expansion}
    f(\cdot)  = \sum_{i=1}^m c_i \phi_i(\cdot),
\end{equation*}
where $c_i \sim \mathcal{N}(0, s(\sqrt{\lambda_i}))$ and the $\phi_i$ and $\lambda_i$ are the corresponding eigenfunctions and eigenvalues of the Laplacian over a compact domain $\mathcal{X}$ with Dirichlet boundary conditions $\phi_i(x) = 0$ on $\partial \mathcal{X}$.

When the domain is the unit hypercube, $\mathcal{X} = [0,1]^d$, the resulting eigenfunctions and eigenvalues can be explicitly computed as 
\begin{equation} \label{eq: dirichlet basis}
    \phi_{j}(x) = 2^{d/2}\prod_{k=1}^d \sin\left(\pi j_k x_k\right), \qquad \lambda_{j} = \sum_{k=1}^d \left(\pi j_k\right)^2,
\end{equation}
where $j = (j_1,\ldots,j_d) \in \mathbb{Z}_m^d$. 
Taking $m$ sinusoidal functions in each dimension yields $m^d$ eigenfunctions in total. For computational purposes, in \texttt{GaussED} the domain $\mathcal{X}$ of the Gaussian model is taken as a $d$-dimensional Cartesian product of intervals $[a_1,b_1]\times\ldots\times[a_d,b_d]$. The required eigenfunctions can be obtained by a simple rescaling of the previous formula.

\section{Computational Details of \texttt{GaussED}} \label{sec: gaussed details}

In this section we provide details of certain aspects of the computational approaches of \texttt{GaussED}. In \Cref{subsec: conditioning}, we derive the relevant conditional distributions of the spectral Gaussian process model detailed in \Cref{subsec: spectral}, under general linear information. In 

\subsection{Conditioning} \label{subsec: conditioning}

In this section, we both derive and discuss \texttt{GaussED}'s approach to conditioning and sampling from the posterior. For completeness, we present the derivation of the conditional distributions of the Gaussian process model detailed in \Cref{app: eigenfunctions}. For the sake of generality we consider a general truncated basis model, which takes the form of
\begin{equation*}
    f(\cdot) = \mu + \sum_{i=1}^{m} c_i \phi_i(\cdot),
\end{equation*}
where the $c_i$ are pairwise independent Gaussian variables and the $\phi_i$ form our basis functions. Suppose that we have a vector of $n$ continuous linear functionals $\bm{\delta}_n = (\delta_1,\ldots,\delta_n)^\top \in \mathcal{D}^n$, such that each $\delta_i$ belong to the design set $\mathcal{D}$ (see \Cref{subsec: bayesian experimental design}). We form the conditional distribution $f\,|\, \bm{\delta}_n(\mathrm{f})$ as follows, letting $c = \left(c_1,\ldots,c_m\right)^\top$, we have 
\begin{equation*}
    \begin{pmatrix}
    c \\
    \bm{\delta}_n(f)
    \end{pmatrix} \sim \mathcal{N}\left(0, \begin{pmatrix}
    K_{cc} & K_{c\bm{\delta}} \\
    K_{\bm{\delta} c} & K_{\bm{\delta}\bm{\delta}}
    \end{pmatrix}\right),
\end{equation*}
where $K_{cc} = \Cov(c,c) \in \mathbb{R}^{m\times m}$, $K_{c\bm{\delta}} = \Cov(c,\bm{\delta}_n) \in \mathbb{R}^{m\times n}$, $K_{\bm{\delta} c} = K_{c\bm{\delta}}^\top$ and $K_{\bm{\delta} \bm{\delta}} = \Cov(\bm{\delta}_n(f),\bm{\delta}_n(f)) \in \mathbb{R}^{n\times n}$. The conditional distribution can be computed using standard finite-dimensional formulae as $ c \,|\,\bm{\delta}_n (f) = \bm{\delta}_n(\mathrm{f}) \sim \mathcal{N}\left(\mu_{\bm{\delta}}, \Sigma_{\bm{\delta}}\right)$,
where
\begin{align}
    \mu_{\bm{\delta}} &= K_{c\bm{\delta}} K_{\bm{\delta} \bm{\delta}}^{-1} \bm{\delta}_n(\mathrm{f}), \label{eq: post mean} \\
    \Sigma_{\bm{\delta}} &= K_{cc} - K_{c\bm{\delta}}K_{\bm{\delta} \bm{\delta}}^{-1}K_{\bm{\delta} c}. \label{eq: post covariance matrix}
\end{align}
Since the components of $c$ are pairwise independent, we have $K_{cc} = \Lambda = \text{diag}\left(\Var(c_1),\ldots,\Var(c_m)\right)$. Furthermore, since $\bm{\delta}_n$ is a vector of linear functionals, we have, for each $i\in \{1,\ldots,n\}$, that $\delta_i f = \sum_{j=1}^m c_j \delta_i \phi_j$.
Therefore, we have
\begin{align*}
    \Cov(\delta_i f, \delta_j f) &= \Cov\left( \sum_{k=1}^m c_k \delta_i \phi_k, \sum_{k=1}^m c_k \delta_j \phi_k \right) 
    = \sum_{k=1}^m\Var(c_k) \delta_i\phi_k \delta_j \phi_k
\end{align*}
and so $K_{\bm{\delta} \bm{\delta}} = (\bm{\delta} \Phi) \Lambda (\bm{\delta} \Phi)^\top$, where $(\bm{\delta} \Phi)_{ij} = \delta_i \phi_j$. Finally, we have
\begin{align*}
    \Cov(\delta_i f, c_j) &= \Cov\left( \sum_{k=1}^m c_k \delta_i \phi_k, c_j\right) 
    = \Var(c_j)\delta_i \phi_j 
\end{align*}
and so $K_{\bm{\delta} c} = (\bm{\delta} \Phi)\Lambda$.
Thus all the required quantities can be explicitly evaluated.

\subsection{Sampling} \label{subsec: sampling}

To sample from the posterior process $f(\cdot) \,|\, \bm{\delta}_n(\mathrm{f})$, we can sample from the conditional distribution $c\,|\,\bm{\delta}_n(\mathrm{f})$ and then utilise the basis expansion of $f$ in \eqref{eq: truncated expansion}. 
To achieve this, we are required to perform a matrix square root of the posterior covariance matrix $\Sigma_{\bm{\delta}}$, and we recall that, when conditioning on exact information, the resulting $\Sigma_{\bm{\delta}}$ is singular in general. 
The standard solution of performing a singular value decomposition (SVD) is unsuitable, since the $\Sigma_{\bm{\delta}}$ often have repeated singular values, which are incompatible with existing implementations of automatic differentiation that assume uniqueness of the singular values \citep{papadopoulo2000SVD, paszke2019pytorch}. 
Although there have been recent efforts to address this \citep{Wang_2021}, the resulting algorithms are computationally prohibitive in our setting. 

An alternative method of sampling from $f(\cdot) \,|\, \bm{\delta}_n(\mathrm{f})$ is called \emph{Matheron's update rule} \citep[][Corollary 4]{wilson2021pathwise}. 
Matheron's update rule takes the form
\begin{equation} \label{eq: matheron}
    f(\cdot) \,|\, \bm{\delta}_n(\mathrm{f}) \overset{d}{=} f(\cdot) + \Cov(f(\cdot), \bm{\delta}_n(f))K_{\bm{\delta}\bm{\delta}}^{-1}(\bm{\delta}_n(\mathrm{f}) - \bm{\delta}_n(f)).
\end{equation}
The advantage of Matheron's update rule over the preceding approach is that we are not required to compute the square root of $\Sigma_{\bm{\delta}}$; this is the default approach used in \texttt{GaussED}.  


\subsection{Optimising the Acquisition Function} \label{subsec: optimising acquisition}

As discussed in \Cref{subsec: stochastic optimisation}, we utilise stochastic optimisation methodology to optimise the acquisition function. Unfortunately, the acquisition functions often exhibit multiple local optima, implying that it is unlikely that the optimiser will find a global optima. 
There are many approaches to reduce this probability, for instance by running the optimiser at different initialisations in parallel. In \texttt{GaussED}, the default approach is to sample uniformly from the design set, then evaluate the acquisition function at each of the sample points, before proceeding to initialise the optimiser at the best obtained point (i.e. Monte Carlo optimisation is used to initialise a stochastic optimisation method). This was the approach used in all the experiments of \Cref{sec: demonstrations}. 

Since our design sets are based on intervals\footnote{Recall from \Cref{sec: demonstrations} that all of the design sets were parameterised as a Cartesian product of intervals.}, we perform a standard reparameterisation to obtain a global optimisation problem in $\mathbb{R}^d$. 
This is achieved by using a scaled logistic function of the form
\begin{equation*}
    \text{logit}(x; a,b) = \log((x-a) / (b-a)) - \log(1-(x-a) / (b-a)),
\end{equation*}
where, for $x,a,b \in\mathbb{R}^d$, we consider $\text{logit}$ to be applied component-wise.

\section{Experimental Details} \label{sec: experiment details}

In this section we present full details for the experiments presented in \Cref{sec: demonstrations}.
All experiments can be reproduced using source code available at \url{https://github.com/MatthewAlexanderFisher/GaussED}.

\subsection{Probabilistic Solution of PDEs} \label{subsec: heat equation details}

\paragraph{Approximating the Loss: }

Following from \Cref{subsec: heat equation}, recall that the quantity of interest was the function $\mathsf{f}$, implying the loss takes the form
\begin{equation*}
    L(g,g') = \| g-g'\|^2 = \int_\mathcal{X} |g(x) - g'(x)|^2 \,\mathrm{d}x. 
\end{equation*}
Since there is not a closed-form solution to this integral when $g$ is a Gaussian process, we proceed by approximating the integral through a cubature rule. For this experiment, we performed a Riemann sum over a uniform $15 \times 15$ grid over the domain $\mathcal{X} = [-1,1]^2$. 

\paragraph{Gaussian Model:}

For this experiment, we used a mean-zero Gaussian process $f$ with Mat\'ern covariance function with smoothness parameter $\nu = 3.5$. 
The Dirichlet boundary conditions of the \ac{PDE} were automatically enforced by the spectral \ac{GP} approximation, applied to the domain $\mathcal{X} = [-1,1]^2$ (c.f. \Cref{eq: boundary conditions}).

\paragraph{Optimisation:}

For both the optimisation of the acquisition function and performing maximum likelihood estimation, we used the Adam stochastic optimisation methodology \citep{kingma2014adam}. 

Using the methodology discussed in \Cref{subsec: optimising acquisition}, at each iteration of SED, we sampled $100$ points uniformly from the design set and computed the corresponding values of acquisition function, using the default values of $N = 81$ and $M = 9$ in the stochastic gradient estimator of \Cref{subsec: stochastic optimisation}. We then proceeded by initialising the stochastic optimiser at the sample point which minimised the acquisition function. The learning rate used was the default value of $10^{-1}$ and the optimiser was run for $1000$ iterations, at each step of SED.

Using the methodology as discussed in \Cref{subsec: hyperparam}, we began optimising the amplitude $\lambda$ and the lengthscale $\ell$ after $n_0 = 10$ iterations of SED. This $n_0 = 10$ is the default value in \texttt{GaussED}. The initial parameter values were taken as the default values of $\lambda = 1$ and $\ell = 0.2$. The learning rate used was the default value of $10^{-3}$ and the optimiser was run $1000$ iterations, at each step of SED.

\paragraph{Code:}
The code used to run the experiment can be seen in \Cref{fig: syntax example} and discussed in \Cref{subsec: heat equation}.

\begin{figure}[ht]
\centering
\begin{minted}[
frame=lines,
framesep=5mm,
baselinestretch=1.2,
fontsize=\footnotesize,
%linenos
]{python}

k = MaternKernel(2, 2, initial_parameters)
domain = [[-1.05,1.05],[-1.05,1.05]]
gp = SpectralGP(k)
gp.set_domain(domain)

exponential_warp = lambda x: torch.exp(3 * x)
qoi = OutputWarp(gp, exponential_warp)()

X,Y = torch.meshgrid(torch.linspace(-1,1,25),torch.linspace(-1,1,25))
mesh = torch.stack([X,Y]).T.reshape(25**2,2)

loss = L2(qoi, mesh)

def d_func(design, m):
    all_phis = []
    for i in range(len(design)):
        design_i = design[i]
        line_int_gps = get_line_int_gps(design_i.unsqueeze(1), gp)
        for j in line_int_gps:
            all_phis.append(j.basis_matrix(None,m))
    return torch.cat(all_phis)
        
def d_sample(design_point, mean, cov, n, random_sample=None):
    all_samples = []
    line_int_gps = get_line_int_gps(design_point.unsqueeze(1), gp)
    matrix_sqrt = gp.solver.square_root(cov)
    for i in line_int_gps:
        samp_i = i.sample(mean, cov, n, random_sample, sqrt=matrix_sqrt)(None)
        all_samples.append(samp_i)
    return torch.cat(all_samples).T

initial_design = torch.Tensor([[0,  0,  0]])
d = Design(d_func, d_sampling, initial_design)
d.set_domain([[0, math.pi],[-1,1],[-1,1]])

acq = BayesRisk(gp, loss, d, nugget=1e-2)
experiment = Experiment(gp, transformed_black_box, d, acq, m=28)
experiment.run(30)
\end{minted}
\caption{The \texttt{GaussED} code used to run the tomographic reconstruction experiment of \Cref{subsec: tomographic reconstruction}.} \label{fig: tomography code}
\end{figure}

\FloatBarrier

\subsection{Tomographic Reconstruction} \label{subsec: tomographic details}

\paragraph{Approximating the Loss:}
Following from \Cref{subsec: tomographic reconstruction}, recall that the quantity of interest was the function $\exp(3\mathsf{f})$, implying the loss takes the form
\begin{equation*}
    L(g,g') = \| \exp(3g)-\exp(3g')\|^2 = \int_\mathcal{X} |\exp(3g(x)) - \exp(3g'(x))|^2 \,\mathrm{d}x. 
\end{equation*}

We follow the same approach of \Cref{subsec: heat equation details} and approximate the integral through a Riemann sum, now over a uniform $25\times 25$ grid over the domain $\mathcal{X} = [-1,1]^2$.

\paragraph{Gaussian Model:} 

For this experiment, we utilised a stationary Gaussian process model $f$ with Mat\'ern covariance with smoothness parameter $\nu = 2$. The Gaussian model is defined on the domain $[-1.05,1.05]^2$, since the boundary conditions of the resulting GP do not necessarily agree with the boundary conditions of the quantity of interest. 



\paragraph{Quantity of Interest:}
Recall from \Cref{subsec: tomographic reconstruction} that the quantity of interest was of the form
\begin{equation*}
    \mathsf{f}(x) = \begin{cases}
    1, &  \text{when }\|x-(0.4,0.4)\| < 0.3, \\
    0, & \text{otherwise}.
    \end{cases}
\end{equation*}
Since this quantity of interest defines a circle within the domain $\mathcal{X} = [-1,1]^2$, it is possible to find a closed form solution to the line integrals of $\mathsf{f}$ for given parameters values $(\theta,x,y)$. However, for ease of implementation and to allow our approach to be easily generalised to more complex examples, we computed the line integrals of $\mathsf{f}$ by performing a Riemann integral over a uniform mesh consisting of $200$ evaluations from $\mathsf{f}$.  

\paragraph{Optimisation: }

All settings used were the same as the previous experiment detailed in \Cref{subsec: heat equation details}, apart from the following settings:

We began optimising the amplitude $\lambda$ and lengthscale $\ell$ parameters of the Gaussian process at step $n_0 = 0$. The initial parameter values were taken as $\lambda = 0.5$ and $\ell = 0.4$. 

\paragraph{Code: }

The \texttt{GaussED} code used to run this experiment is presented in \Cref{fig: tomography code}. The structure of the code is quite different to the code used in the other experiments (\Cref{fig: syntax example} and \Cref{fig: bayesian optimisation code}). This is due to the fact that the design object (\texttt{d}) is not instantiated by the \texttt{EvaluationDesign} class. Note that for both the PDE experiment (\Cref{subsec: heat equation}) and the Bayesian optimisation experiment (\Cref{subsec: lotka volterra}), the design sets $\mathcal{D}$ consisted of evaluations of the Gaussian process, $\delta (f) = f(x)$, or its derivatives $\delta (f) = \partial_i f(x)$. In situations such as these, the \texttt{EvaluationDesign} class may be used.  For this example, however, the observed data consists of line integrals. Therefore, in this more general situation, we must specify two further functions: Given a parameterisation $\mathcal{D}_\theta$ of the design set, the first function must take in a sequence of parameters $\theta_1,\ldots,\theta_n$ and return the corresponding $(\bm{\delta}\Phi)_{ij} = \delta_{\theta_i} \phi_j$ matrix, where the $\phi_j$ are the eigenfunctions of \eqref{eq: dirichlet basis}. This is reflected in the code (\Cref{fig: tomography code}) in the function \texttt{d\_func}, which, for each design set parameter constructs the corresponding line integral for a given number of basis functions (\texttt{m}). The second function we must specify must be able to, given a parameter $\theta$, sample from the process $\delta_{\theta} f\,|\, \bm{\delta}_n$, where $\bm{\delta}_n$ is data gathered from SED. This is directly reflected in the code (\Cref{fig: tomography code}) in the function \texttt{d\_sample}. Note that, in \Cref{fig: tomography code} we omit the \texttt{get\_line\_int\_gps} function. This is a function that, given a parameter value $\theta$ and Gaussian process $f$, returns the corresponding $\delta_\theta f$ object. We do this because \texttt{get\_line\_int\_gps} is complexified due to the parameterisation of the line function $r(x)$ and the calculation of the limits of integration $a,b$ in the line integral
\begin{equation*}
    \int_a^b f(r(x))\,\mathrm{d}x.
\end{equation*}
We, therefore, omit \texttt{get\_line\_int\_gps} for clarity.

A second major difference, is the use of an output warp (\texttt{OutputWarp}). Due to the non-linear nature of the output warp, the resulting object \texttt{qoi} is only able to sample from the prior and posterior. Note that the syntax for specifying a output deformation of a GP is the same as specifying other transformations (e.g. see \Cref{fig: syntax example} and \Cref{fig: bayesian optimisation code}).

Another difference is that the Gaussian model specified in the PDE experimental code (\Cref{fig: syntax example}) agreed with the boundary conditions of the PDE; here, however, we specify the domain (\texttt{gp.set\_domain}) as $[-1.05,1.05]\times[-1.05,1.05]$. 
Since we took the domain of the Gaussian process to be larger than the domain on which the task is defined, we must also specify the domain of the design object (\texttt{d.set\_domain}), which otherwise, by default, would be taken as the same the Gaussian model (\texttt{gp}).

Finally, note that the acquisition function (\texttt{acq}), as discussed previously, is instantiated with a nugget value of $10^{-2}$ and the experiment object (\texttt{experiment}) is instantiated with $m = 28^2$ basis functions. This is in contrast to the code for the PDE example (\Cref{fig: syntax example}), which used the default value of $m=30^2$ basis functions.

\subsection{Gradient-Based Bayesian Optimisation} \label{subsec: lotka-volterra details}

\paragraph{Approximating the Loss:}

Recall from \Cref{subsec: lotka volterra} that our quantity of interest is $q(\mathsf{f}) = \max_{x\in \mathcal{X}} \log \mathcal{L}(x)$. Thus, our loss function takes the form
\begin{equation*}
    L(g,g') = \left| \max_{x \in \mathcal{X}}\left(g(x)\right) - \max_{x \in \mathcal{X}}\left(g'(x)\right)\right|^2. 
\end{equation*}

In order to optimise the samples, we used a grid-based optimiser using a uniform $40\times 40$ grid over the domain of interest $[0.45,0.9]\times[0.09,0.5]$.

\paragraph{Gaussian Model:}
For this experiment, we used a mean-zero stationary Gaussian model $f$ with Mat\'ern covariance, with smoothness parameter $\nu = 3$.
Since our \ac{GP} satisfies the boundary conditions in \eqref{eq: boundary conditions}, which are unrelated to the task at hand, we took the domain of the GP to be  $[0.4,0.95]\times[0.04,0.55]$, which is wider than the domain on which the task is defined. 

\paragraph{Quantity of Interest:}

Synthetic data $y = (p_i,q_i)_{i=1}^{51}$ were generated at times $t = 0, 0.5,1,\ldots,50$ by perturbing the solution of the Lotka--Volterra model, with parameter values $(\alpha,\beta,\gamma,\delta)= (0.5,0.1,0.3,0.1)$, with mean-zero Gaussian errors with variance $\sigma^2 = 0.05^2$. The data used for the log-likelihood and the corresponding true solution with $(\alpha,\beta,\gamma,\delta) = (0.5,0.1,0.3,0.1)$ are displayed in \Cref{fig: bod data}. 

\begin{figure}[h!] 
   \centering
   \includegraphics[width=0.75\textwidth]{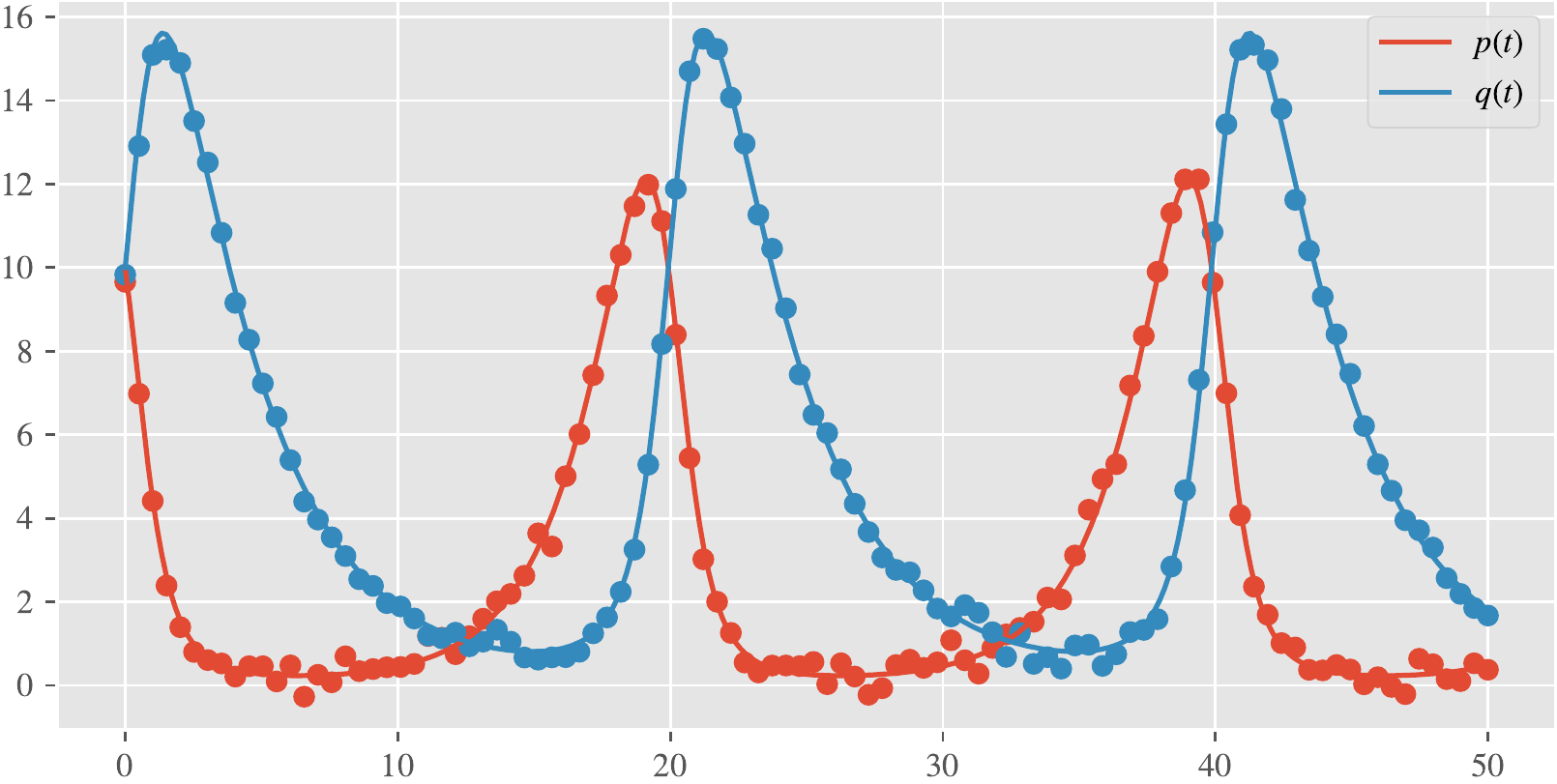}

\caption{
    Solution of the Lotka--Volterra model with parameter values $\theta = (0.5,0.1,0.3,0.1)$, with the synthetic data $y = (p_i,q_i)_{i=1}^{51}$ overlaid.
}
\label{fig: bod data}
\end{figure}

\paragraph{Optimisation: }

All settings were as the previous experiment detailed in \Cref{subsec: heat equation details}, apart from the following settings:

For both the optimisation of the acquisition function and performing maximum likelihood estimation, we used the Adam stochastic optimisation methodology \citep{kingma2014adam}. 

Using the methodology as discussed in \Cref{subsec: optimising acquisition}, at each iteration of SED, we sampled $100$ points times uniformly from the design set and computed the corresponding values of acquisition function, using the default values of $N = 81$ and $M = 9$ in the stochastic gradient estimator of \Cref{subsec: stochastic optimisation}. We then proceeded by initialising the stochastic optimiser at the sample point which minimised the acquisition function. The learning rate used was the default value of $10^{-1}$ and the optimiser was run $1000$ iterations, at each step of SED. Furthermore, in order to increase the numerical stability of linear algebra operations, we used a nugget term of value $10^{-5}$.

For this experiment, we began optimising the amplitude $\lambda$ and lengthscale $\ell$ at step $n_0 = 10$. The initial kernel parameter values were taken as the values of $\lambda = 1$ and $\ell = 0.1$. The initial parameter values of the spatial deformation were taken as $\theta_x = (1,0,0,1)$ and $\theta_y = (1,0,0,1)$, thus specifying the initial spatial deformation as the identity function. The learning rate used was the default value of $10^{-3}$ and the optimiser was run $1000$ iterations, at each step of SED.

\paragraph{Code:}

The \texttt{GaussED} code used to run this experiment is presented in \Cref{fig: bayesian optimisation code}. The structure of the program is very similar in nature to the PDE experiment of \Cref{subsec: heat equation}. The first difference is that, at each step of SED, we evaluate multiple functionals $\delta$ from the design $\mathcal{D}$. This is directly reflected in \Cref{fig: bayesian optimisation code}, where the design object (\texttt{d}) is constructed by the statistical model $f$ and its first derivatives (\texttt{[gp, gp\_d1, gp\_d2]}).  

The second difference is that, at each step SED, we perform a maximisation, rather than an integral, of sample paths when estimating the acquisition function. In the code for the PDE experiment (\Cref{fig: syntax example}) the integral of posterior samples is hidden within the loss object (\texttt{L2(qoi)}), which, by default, performs a Riemann sum over a uniform mesh if the quantity of interest \texttt{qoi} is function valued. Therefore, in \Cref{fig: bayesian optimisation code} we specify a numerical method that acts on samples from $f$. In this instance, we perform a grid search (\texttt{maximise\_method}) over a uniform $40\times 40$ mesh (\texttt{mesh}) over the domain of optimisation.

Another difference is that the Gaussian model specified in the PDE experimental code (\Cref{fig: syntax example}), agrees with the boundary conditions of the PDE and therefore the domain of the GP is taken as the default value $[-1,1]^2$. In \Cref{fig: bayesian optimisation code}, we must specify the domain (\texttt{gp.set\_domain}) as $[0.4,0.95]\times[0.04,0.55]$. Since we took the domain of the Gaussian process to be larger than the domain over which we wish to maximise, we must also specify the domain of the design object (\texttt{d.set\_domain}), which otherwise, by default, would be taken as the same the Gaussian model (\texttt{gp}).

The final difference is that, in order to increase the numeric stability of linear algebra operations in the SED, we specify a nugget term (\texttt{nugget}) of value $10^{-5}$ in the acquisition function \texttt{acq}.

\begin{figure}[ht]
\centering
\begin{minted}[
frame=lines,
framesep=5mm,
baselinestretch=1.2,
fontsize=\footnotesize,
%linenos
]{python}
k = MaternKernel(3, 2, initial_parameters)
gp = SpectralGP(k)
gp.set_domain(torch.Tensor([[0.4,0.95],[0.04,0.55]]))

gp_d1 = Differentiate(gp,[0],[1])()
gp_d2 = Differentiate(gp,[1],[1])()

x, y = torch.meshgrid(torch.linspace(0.45,0.9,40),torch.linspace(0.09,0.5,40))
mesh = torch.stack([x,y]).T.reshape(40**2,2)
maximise_method = GridSearch(mesh)

qoi = Maximise(gp, maximise_method)()

d = EvaluationDesign([gp, gp_d1, gp_d2], initial_design=torch.Tensor([[0.675, 0.295]]))
d.set_domain(torch.Tensor([[0.45,.9],[0.09,0.5]]))

loss = L2(q)
acq = BayesRisk(q, loss, d, nugget=1e-5)

experiment = Experiment(gp, lotka_volterra, d, acq, m=35)
experiment.start_hyp_optimising_step = 10
experiment.run(30)
\end{minted}
\caption{The \texttt{GaussED} code used to run the gradient-based Bayesian optimisation experiment of \Cref{subsec: lotka volterra}.} \label{fig: bayesian optimisation code}
\end{figure}

\section{Evaluating Computational Aspects of \texttt{GaussED}} \label{sec: evaluating Gaussed}

In this section we empirically investigate computational aspects of \texttt{GaussED}. In \Cref{subsec: stochastic optimisation investigation}, we explore the role of the optimisation methodology and how this affects the experimental design as well as the quality of output. In \Cref{subsec: basis function investigation}, we investigate how the number of basis functions used, for a given problem, affects the quality of posterior inference.

\subsection{Investigating the Efficacy of Stochastic Optimisation} \label{subsec: stochastic optimisation investigation}

In this section, we investigate the effect of the random seed on the quality of the experimental design and, further, investigate the effect of changing the stochastic optimisation approach itself. 
To explore these aspects of \texttt{GaussED}, we repeat the Bayesian optimisation with gradient data experiment presented in \Cref{subsec: lotka volterra}. Recall that, in all the demonstrations in \Cref{sec: demonstrations}, we utilised the Adam stochastic optimisation method \citep{kingma2014adam}. 

Results on the effect of the random seed can be seen in \Cref{fig: lotka_volterra initial error} and \Cref{fig: lotka_volterra initial designs}. The obtained designs imply that our approach of SED is sensitive to the initial conditions. Although the specific design is sensitive, the overall performance and qualitative nature of the designs are approximately independent of random seed.

\begin{figure}[h!]
\centering
\begin{subfigure}[b]{0.4\linewidth}
\includegraphics[width=\textwidth]{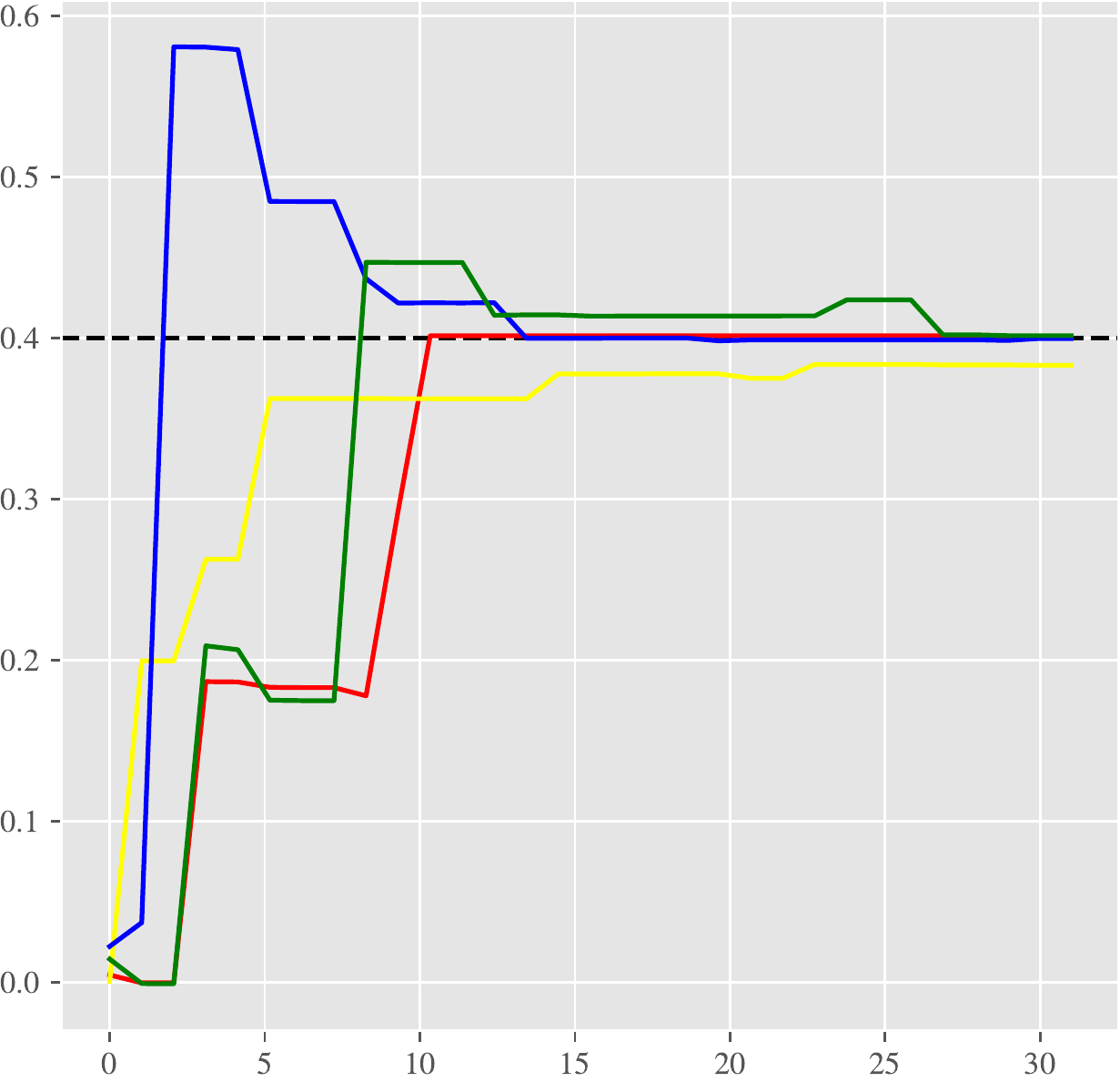} 
\caption{}
\end{subfigure}
\begin{subfigure}[b]{0.4\linewidth} 
\includegraphics[width=\textwidth]{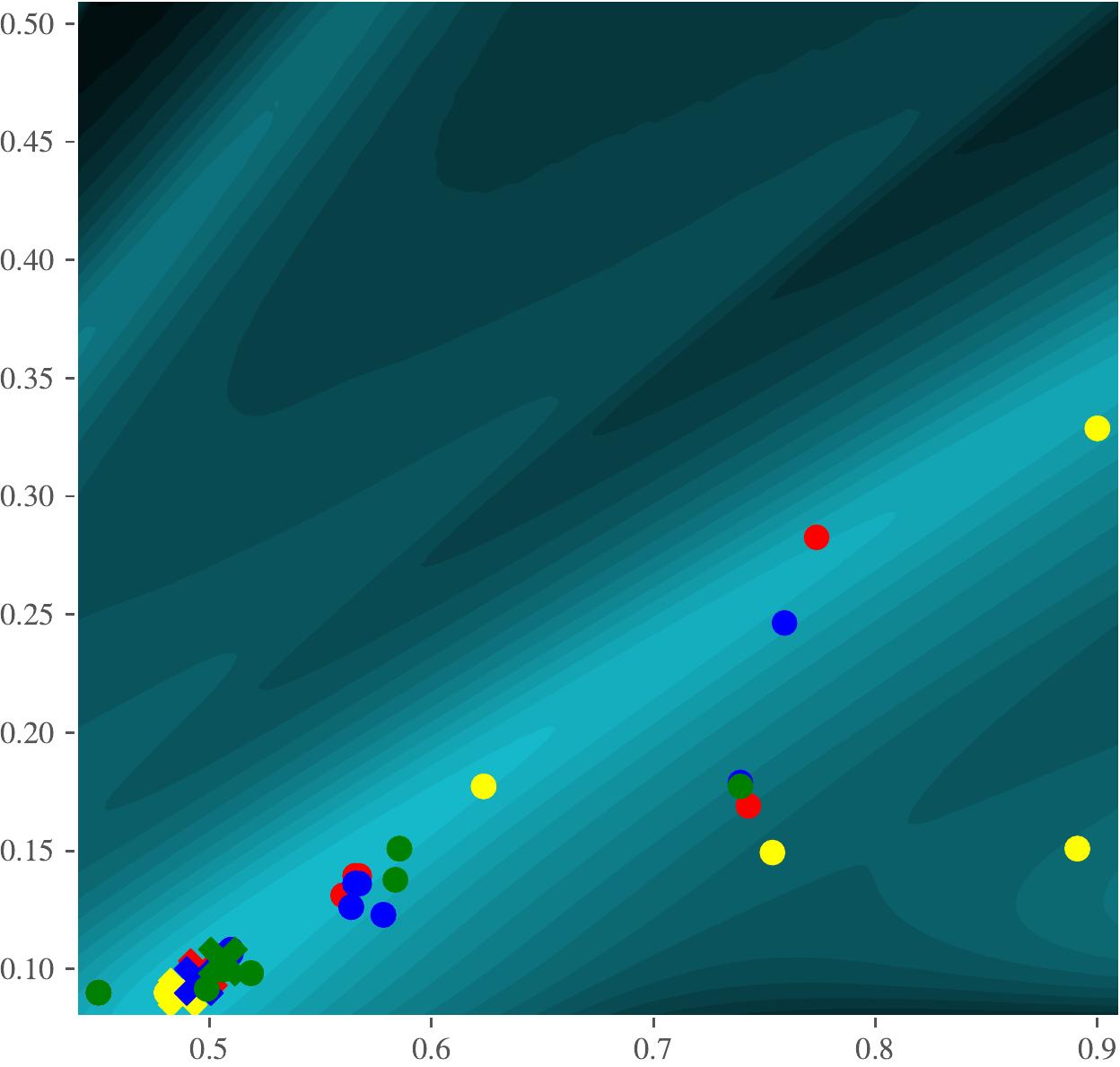}
\caption{} 
\end{subfigure}
\caption{ Convergence analysis of Bayesian optimisation with $4$ different random seeds. The left panel (a) displays the maximal value obtained for each of the random seeds, with each colour corresponding to a different random seed.  The right panel (b) displays the coordinate positions of the obtained maximum value, where the colored symbols \ding{54} indicate the coordinate position of the obtained maximum value at termination. Again, the maximum of the posterior mean is reported.
}
\label{fig: lotka_volterra initial error}   
\end{figure}

\begin{figure}[h!]
\centering

\includegraphics[width=0.6\textwidth]{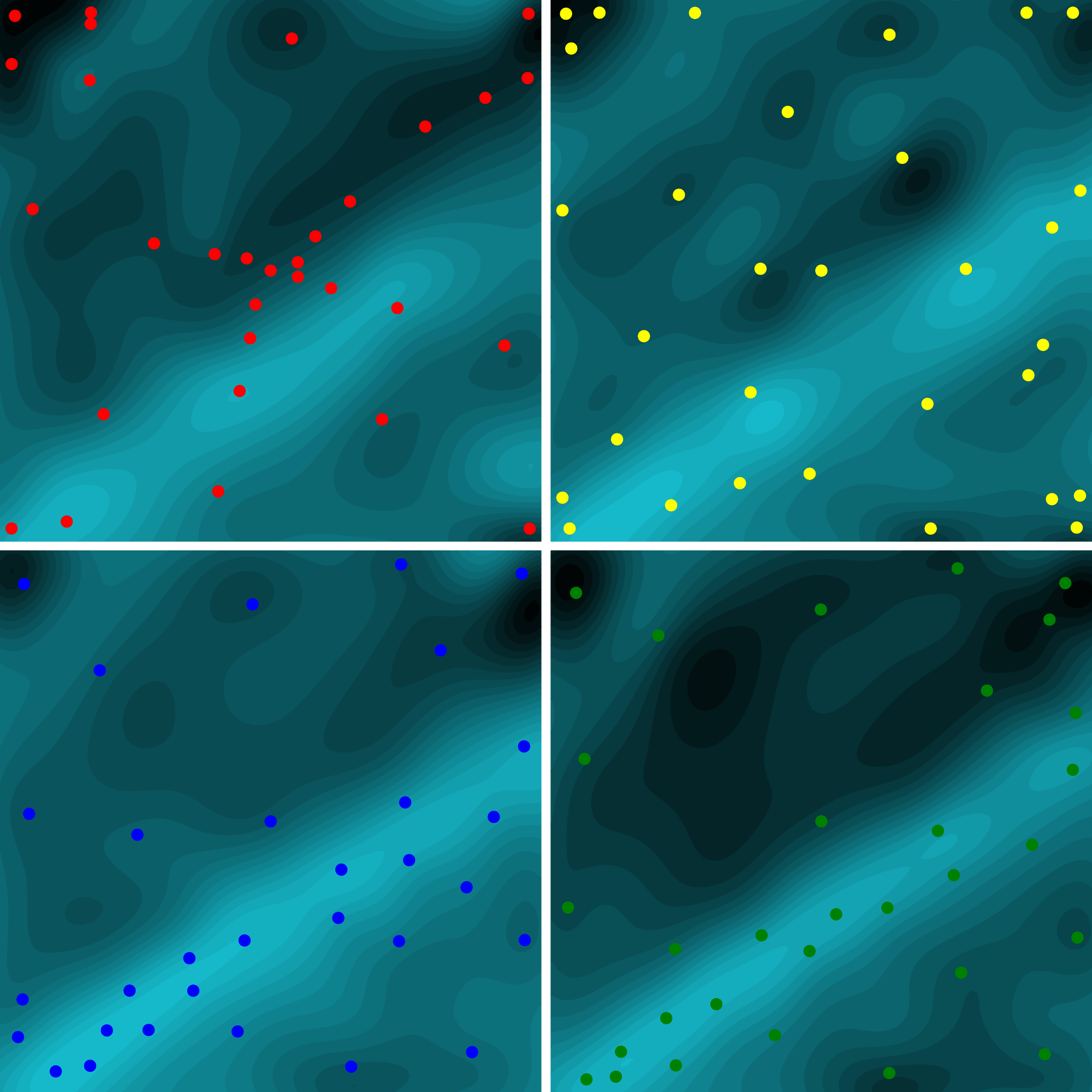} 

\caption{ Designs obtained by SED for the $4$ different random seeds along with the corresponding obtained posterior means. The colours correspond to the same random seed as displayed in \Cref{fig: lotka_volterra initial error}.
}
\label{fig: lotka_volterra initial designs}   
\end{figure}

Results on the effect of stochastic optimisation methodology can be seen in \Cref{fig: lotka_volterra method error} and \Cref{fig: lotka_volterra method designs}. In each of these experiments, the random seed was fixed, and so we are only comparing the effect of different optimisation methodologies. In each experiment, the learning rate was set at $10^{-1}$ and the other parameter values were taken as their default values, as specified in \texttt{PyTorch} \citep{paszke2019pytorch}. 

\begin{figure}[h!]
\centering
\begin{subfigure}[b]{0.4\linewidth}
\includegraphics[width=\textwidth]{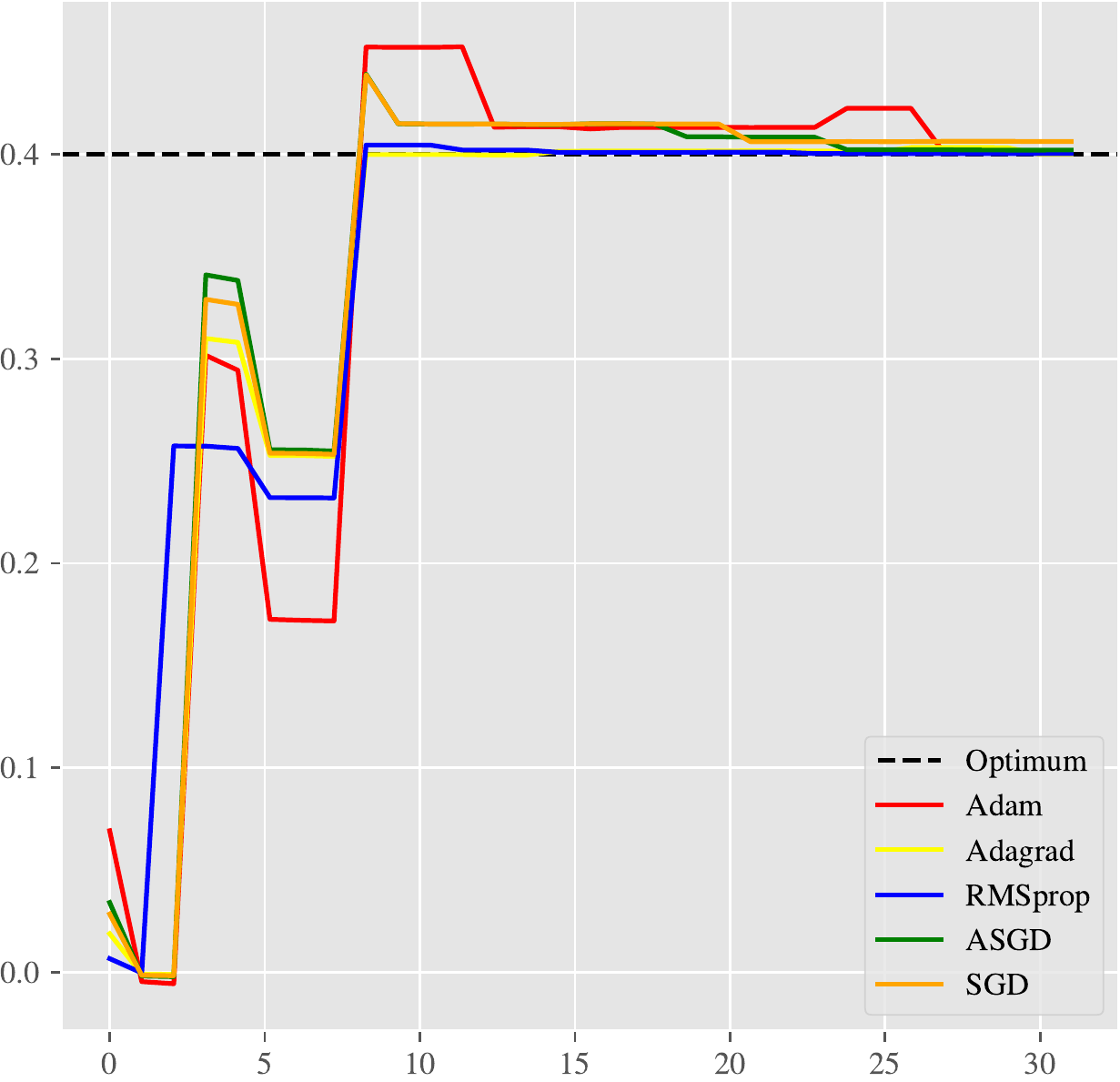} 
\caption{}
\end{subfigure}
\begin{subfigure}[b]{0.4\linewidth} 
\includegraphics[width=\textwidth]{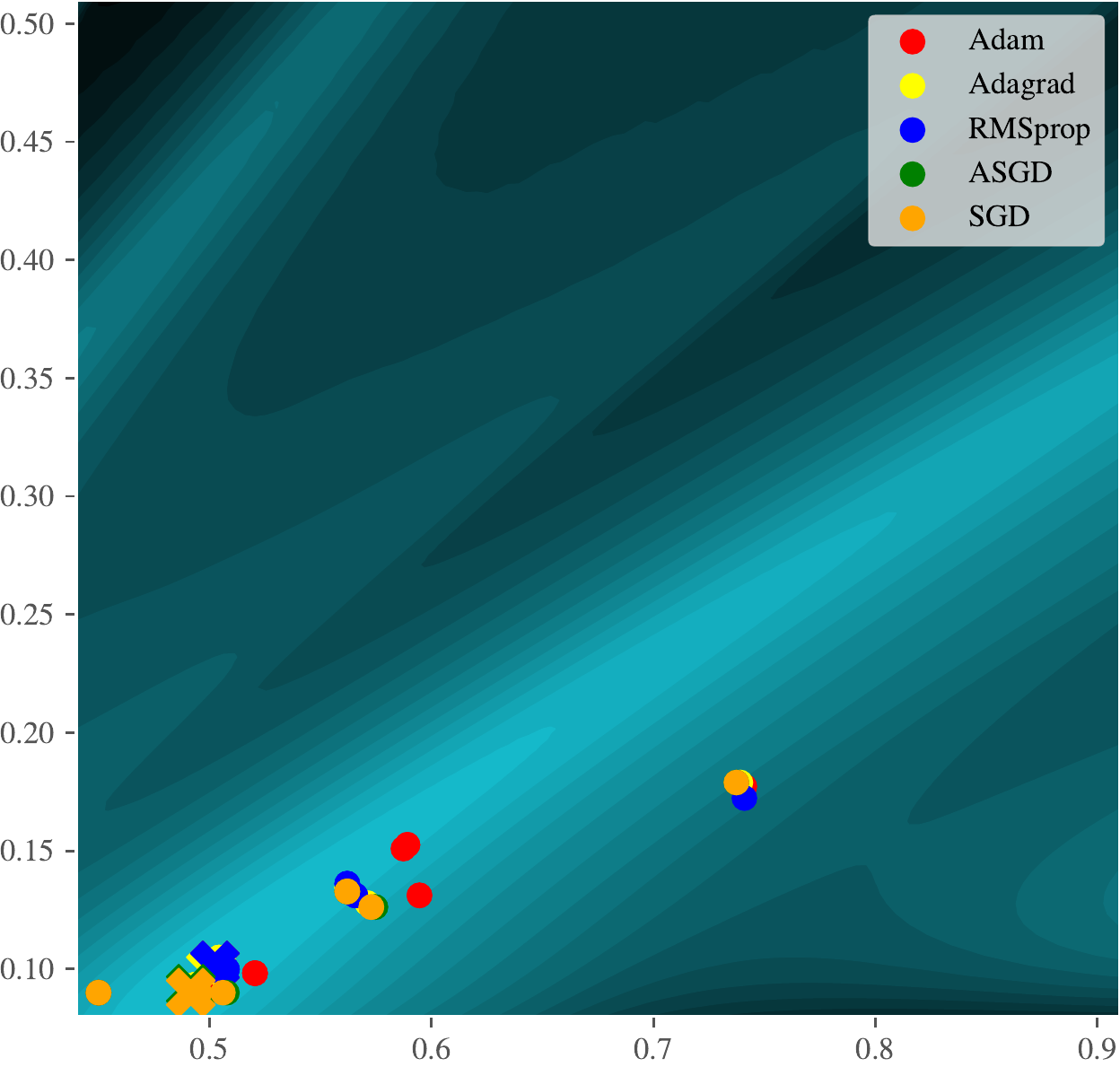}
\caption{}
\end{subfigure}
\caption{ Convergence analysis of Bayesian optimisation with different optimisation methods. The left panel (a) displays the maximal value obtained for each of the optimisation methods, with each colour corresponding to a different method. The right panel (b) displays the coordinate positions of the obtained maximum value, where the colored symbols \ding{54} indicate the coordinate position of the obtained maximum value at termination. Again, the maximum of the posterior mean is reported.
}
\label{fig: lotka_volterra method error}   
\end{figure}

\begin{figure}[h!]
\centering

\includegraphics[width=0.8\textwidth]{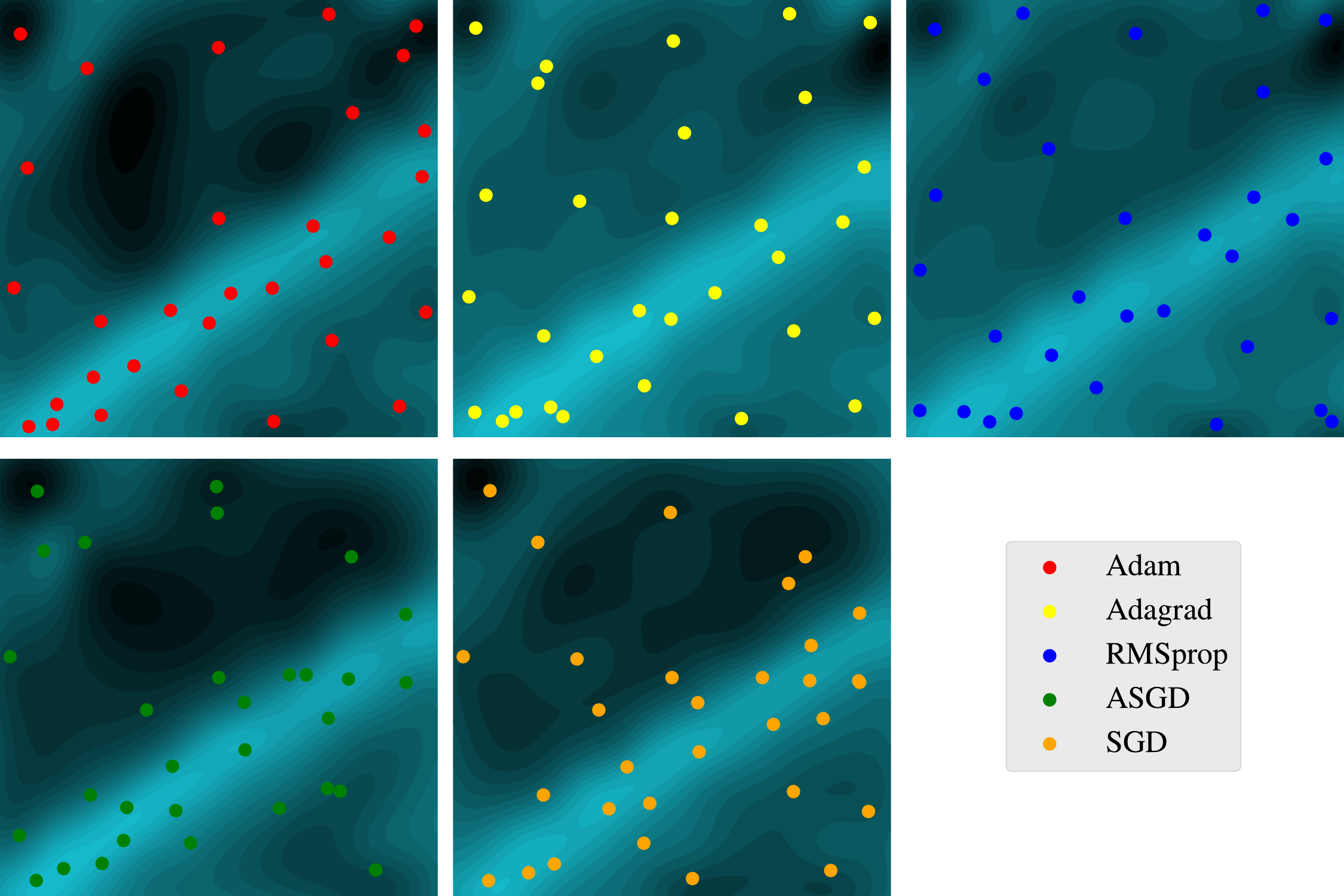} 

\caption{ Designs obtained by SED for the $5$ different optimisation methods along with the corresponding obtained posterior means. The colours correspond to the same optimisation methods as displayed in \Cref{fig: lotka_volterra method error}.
}
\label{fig: lotka_volterra method designs}   
\end{figure}

\subsection{Investigating the Effect of the Number of Basis Functions} \label{subsec: basis function investigation}

Picking an appropriate number of basis functions for a given problem is an important means to reduce computational cost in \texttt{GaussED}. 
In this section, we investigate how the number of basis functions may affect the quality of posterior inference. To this end, it is sufficient to consider the behaviour of posterior sampling in dimension $d=1$, since the behaviour will naturally extend to higher-dimensions due to the exponential scaling of the number of basis function due to \eqref{eq: dirichlet basis}.

In the event where the number of basis functions is smaller than the number of linearly independent data, the resulting posterior will not be well-defined in general. 
The introduction of a nugget term on the diagonal of the covariance matrix, implicitly assuming noisy Gaussian observations, is a pragmatic solution that is widely-used. 
However, the success of this strategy depends crucially on an appropriate amount of regularisation being introduced.

Results on the effect on the number of basis functions and the nugget term are presented in \Cref{fig: matern basis functions}. Through visual inspection, by $m=20$ basis functions, it appears that the posterior process has converged sufficiently well to the true posterior process. Note that, when $m = 7$, the posterior sample paths overlap. This is due to there being only one value of $c_1,\ldots,c_7$ such that the truncated basis model agrees with the $7$ evaluations.

\begin{figure}[h!]
\centering

\includegraphics[width=\textwidth]{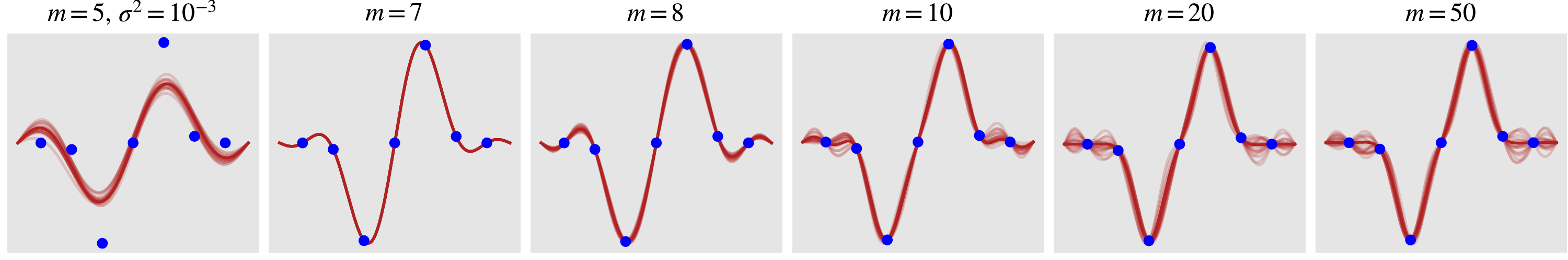} 

\caption{ Samples and posterior mean based on a mean-zero Gaussian process $f$ with Mat\'ern covariance with smoothness parameter $\nu = 1.5$, amplitude $\lambda = 0.1$ and lengthscale $\ell = 0.1$, conditioned to interpolate the $7$ (blue) data points. The corresponding number $m$ of basis functions used in each experiment is displayed in the titles of the subplots. In the event where $m$ is smaller than the number of data points conditioned upon, the corresponding nugget term $\sigma^2$ is also displayed.
}
\label{fig: matern basis functions}   
\end{figure}

\end{document}